\documentclass[11pt]{article}

\usepackage[margin=1in]{geometry}
\usepackage{amsmath,amssymb,amsthm,mathtools}
\usepackage{enumitem}
\usepackage{booktabs}
\usepackage{microtype}
\usepackage{hyperref}

\newtheorem{theorem}{Theorem}[section]
\newtheorem{lemma}[theorem]{Lemma}
\newtheorem{corollary}[theorem]{Corollary}

\newtheorem{definition}[theorem]{Definition}
\newtheorem{assumption}[theorem]{Assumption}
\newtheorem{remark}[theorem]{Remark}

\DeclareMathOperator{\OPT}{OPT}
\DeclareMathOperator{\SOL}{SOL}
\DeclareMathOperator{\dist}{dist}
\DeclareMathOperator{\diam}{diam}
\DeclareMathOperator{\poly}{poly}
\DeclareMathOperator{\polylog}{polylog}

\newcommand{\cB}{\mathcal B}
\newcommand{\cN}{\mathcal N}

\title{Average-Case Minimum-Weight Decoding\\ of Two-Dimensional Color Codes}
\author{Daniel Gibney \and Jackson Huffstutler \and Kaamil Kaka}
\date{}

\begin{document}
\maketitle

\begin{abstract}
Color codes are an important family of topological quantum error-correcting
codes, in part because they support a richer set of transversal logical gates
than surface codes. Their practical use, however, requires efficient decoding.
Minimum-weight decoding is a particularly natural objective under independent
noise, but its computational complexity was unresolved until Walters and Turner
recently proved that minimum-weight decoding for the two-dimensional color code
is NP-hard. Subsequent independent work of Walters and of Gu, Wang, and Kubica showed that
this worst-case hardness can nevertheless be circumvented approximately by
giving polynomial-time approximation schemes.  In terms of the lattice size $N$,
the algorithm of Walters runs in
\(O(N^{2+528/\epsilon})\) time, while the algorithm of Gu, Wang, and Kubica
runs in
\(N(\log N)^{O(1/\epsilon)}\) time for a
\((1+\epsilon)\)-approximation.

We ask whether stronger guarantees are possible on the random instances arising
from a natural iid noise model, and show that they are. For triangular
color-code lattices with $N$ faces and iid face-error probability $p$, our
algorithms have the following high-probability guarantees. If
\(p=o(1/\log^2 N)\), they compute an exact minimum-weight correction in
\(N^{1+o(1)}\) deterministic time. If \(p=\Omega(1/\log^2 N)\), with $p$
bounded away from $1/2$, they compute a \(1+o(1)\)-approximation in the same
near-linear deterministic time. When the syndrome is supplied sparsely, with
\(s\) denoting its observed size, the corresponding running times on the same
high-probability events improve to \((1+s)N^{o(1)}\) in the exact sparse regime
and \(s^{1+o(1)}\) above it.

The algorithms exploit different forms of random structure at different noise
scales. In the sparse regime, the syndrome localizes into small regions that can
be solved exactly. At intermediate densities, a simple peeling step removes
locally unambiguous errors and leaves only a much smaller residual problem. At
larger densities, a block decomposition incurs only a small additive loss,
which becomes negligible relative to the typical optimum. Together, these
results give a near-linear average-case counterpart to the recent worst-case
hardness and approximation results for color-code decoding.
\end{abstract}

\section{Introduction}

Let $G=(V,F)$ be a triangular color-code lattice with
\(
N:=|V|
\)
vertices and $|F|=\Theta(N)$ triangular faces.  An error is a set of faces
$E\subseteq F$.  Its syndrome is the set of vertices incident to an odd number
of faces in $E$.  Given a valid syndrome $S$, we study the minimum-size
syndrome realization problem
\[
\OPT(S)
:=
\min\{|C|:C\subseteq F,\ S(C)=S\}.
\]
The input syndrome is generated by iid face noise: every face belongs to the
physical error independently with probability $p=p(N)$.

The worst-case decoding problem is computationally difficult, but this does not
preclude substantially faster algorithms on the random instances produced by
local noise.  The main point of this paper is that the algorithmic structure
changes as $p$ varies.  No single proof mechanism is responsible for all of the
results.

\subsection{Main results}

We now summarize the formal guarantees.

\begin{theorem}[Sparse exact regime, informal]
\label{thm:intro-sparse}
If
\(
p=o\left(\frac1{\log^2N}\right),
\)
then there is a decoder that outputs a globally minimum-size correction with
high probability.  Its deterministic capped running time is
\(
N^{1+o(1)}.
\)
If the syndrome is supplied sparsely, then with high probability its running
time is
\[
(1+s)N^{O(1/\Delta_N)}\operatorname{polylog}N
=
(1+s)N^{o(1)},
\]
where
\[
\Delta_N
=
\min\left\{
\log\frac1{p\log^2N},
\log N
\right\}.
\]
\end{theorem}

\begin{theorem}[Singleton-peeling regime, informal]
\label{thm:intro-transition}
If
\[
Np^2\to\infty
\qquad\text{and}\qquad
\frac{p\log N}{\log(1/p)}\to0,
\]
then there is a syndrome-only decoder that, with high probability, outputs a
correction $C$ satisfying
\(
|C|
\le
(1+O(p))\OPT(S).
\)
Its deterministic capped running time is
\(
N^{1+o(1)}.
\)
If the syndrome is supplied in sparse form, then with high probability its
running time is
\(
s^{1+o(1)}.
\)
\end{theorem}

\begin{theorem}[Average-case approximation tradeoff, informal]
\label{thm:intro-dense}
Assume every vertex is incident to exactly $d=O(1)$ faces, and let
\[
\alpha(p):=\frac{1-(1-2p)^d}{2}.
\]
For every $\epsilon>0$, provided $N\alpha(p)\to\infty$, the block decoder can be
parameterized so that, with probability
\(
1-\exp(-\Omega(N\alpha(p))),
\)
it outputs a correction $C$ satisfying
\(
|C|
\le
(1+\epsilon)\OPT(S).
\)
Its deterministic running time is
\(
N2^{O(1/(\epsilon\alpha(p)))}.
\)
Moreover, on the same high-probability event, its running time can be bounded
in terms of the observed syndrome size by
\[
O\left(
\frac{s}{\alpha(p)}
2^{O(1/(\epsilon\alpha(p)))}
\right).
\]
In particular, for constant $p$, with probability
\(
1-\exp(-\Omega(N)),
\)
the running time is
\(
s2^{O(1/\epsilon)}.
\)
\end{theorem}

\begin{corollary}[Near-linear near-optimal decoding]
\label{cor:intro-dense-nearlinear}
If
\(
p\log N\to\infty,
\)
then the block decoder can be parameterized so that, with high probability,
it outputs a correction satisfying
\(
|C|
\le
(1+o(1))\OPT(S).
\)
Its deterministic running time is
\(
N^{1+o(1)}.
\)
Moreover, on the same high-probability event, its running time can be bounded
in terms of the observed syndrome size by
\(
s^{1+o(1)}.
\)
\end{corollary}

The three guarantees overlap sufficiently to cover the full range of noise
rates considered here.  Below the $1/\log^2 N$ scale, the sparse decoder is
exact with high probability.  At and above that scale, either singleton
peeling or block decoding gives a $(1+o(1))$ approximation in near-linear
time.  In particular, there is no asymptotic gap between the peeling and block
regimes.  The precise crossover argument is given in
Section~\ref{sec:combined}.

The resulting guarantee is:
\begin{center}
\resizebox{\textwidth}{!}{
\begin{tabular}{llll}
\toprule
Noise regime (iid face noise) &
Approximation (w.h.p.) &
Sparse-input time (w.h.p.) & Worst-case time\\
\midrule
$p=o(1/\log^2N)$
& exact
& $(1+s)N^{o(1)}$
& $N^{1+o(1)}$\\
$p=\Omega(1/\log^2N)$
& $1+o(1)$
& $s^{1+o(1)}$
& $N^{1+o(1)}$\\
\bottomrule
\label{table:two_regime}
\end{tabular}
}
\end{center}

All rows retain a deterministic $N^{1+o(1)}$ capped worst-case running time.
The second row suppresses stronger regime-specific information.  On the
lower-density side, singleton peeling gives the explicit approximation ratio
$1+O(p)$ and leaves a residual syndrome of size $O(sp)$.  On the higher-density
side, the block decoder gives the quantitative tradeoff
\[
T
=
O\left(
\frac{s}{\alpha(p)}
2^{O(1/(\epsilon\alpha(p)))}
\right)
\]
on the syndrome-density event.

\begin{remark}[Success probabilities]
\label{rem:intro-success-probabilities}
The probability guarantees can be made more explicit.
For the block approximation regime, the approximation guarantee holds with
failure probability
\(
\exp(-\Omega(N\alpha(p))).
\)
For small $p$, this is $\exp(-\Omega(Np))$, and for constant $p$ it is
$\exp(-\Omega(N))$.
For singleton peeling, the local-pair estimate
\(
|E\triangle A|=O(Np^2)
\)
fails with probability $\exp(-\Omega(Np^2))$.  For every fixed
$\gamma>0$, the residual-polymer cutoff can be chosen so that the complete
transition decoder fails with probability at most
\(
\exp(-\Omega(Np^2))+O(N^{-\gamma})
\)
while preserving the $s^{1+o(1)}$ successful-instance running time and the
$N^{1+o(1)}$ deterministic cap.

In the sparse exact regime, for every fixed $\gamma>0$ the constants can be
chosen so that exactness fails with probability at most $O(N^{-\gamma})$ while
the successful-instance running time remains
\(
(1+s)N^{o(1)}.
\)
If one prefers a superpolynomially small failure probability, the separation
radius may be increased by a slowly growing factor while still retaining an
$N^{o(1)}$ local-search overhead.
\end{remark}

The remainder of the paper proves these statements.  We restrict attention to
iid noise and do not consider the smoothed model.

\subsection{Solution overview}

\paragraph{Syndrome-sensitive running time.}
In addition to the lattice size $N$, we write
\(
s:=|S|
\)
for the number of violated syndrome vertices.  This parameter is particularly
natural at low noise, where typically $s\ll N$.  Whenever we state a running
time in terms of $s$, we assume that the syndrome is supplied as a sparse list
of its violated vertices and that local lattice neighborhoods can be explored
in constant time per visited cell.  If instead the syndrome is supplied as an
$N$-bit vector, reading the input alone costs $\Omega(N)$ time.

We distinguish two kinds of running-time statements.  The algorithms can be
capped and equipped with a deterministic block-decoder fallback, giving
worst-case running time $N^{1+o(1)}$.  On the high-probability structural events
used in the average-case analysis, the same algorithms can be implemented by
exploring only neighborhoods of syndrome vertices and their residual
components.  This yields substantially sharper running times in terms of $s$.

\paragraph{Sparse regime: exact geometric localization.}
At the sparsest scales, the syndrome breaks into well-separated local
components.  A global optimum cannot profitably connect two sufficiently
separated components because any connecting correction must cross a long
annulus, while each component has a much cheaper local witness.  Likewise, an
optimum cannot escape far outside one inflated local region.  This gives an
exact decomposition of every global optimum.

For
\(
p=o(1/\log^2N),
\)
the local optimum size can be bounded by
\[
R=O\left(\frac{\log N}{\Delta_N}\right),
\qquad
\Delta_N
:=
\min\left\{
\log\frac1{p\log^2N},
\log N
\right\},
\]
for any fixed desired polynomial success exponent after increasing the hidden
constant.  Since $\Delta_N\to\infty$, bounded-depth branching solves each local
instance in
\(
d^{O(R)}
=
N^{O(1/\Delta_N)}
=
N^{o(1)}
\)
time.

The syndrome-sensitive implementation does not scan all blocks.  There are only
$O(s)$ active blocks, and all relevant inflated regions are discovered from
them.  Thus, on the sparse good event, the total running time is
\[
(1+s)N^{O(1/\Delta_N)}\operatorname{polylog}N
=
(1+s)N^{o(1)}.
\]
A deterministic cap and fallback retain an $N^{1+o(1)}$ worst-case bound.

\paragraph{Transition regime: peel the easy part and solve the ambiguity.}
At larger $p$, exact geometric decomposition of the original syndrome is no
longer available.  Instead, we use a different observation: an isolated
physical single-face error is directly visible in the syndrome.  We peel all
faces whose three incident vertices are violated.  Let
\[
A(S):=\{f:V(f)\subseteq S\}
\]
be the peeled set and
\(
B:=E\triangle A(S)
\)
the discrepancy between peeling and the hidden physical error.

A peeling discrepancy can occur only near a collision of at least two physical
errors.  If $X(E)$ denotes the number of constant-distance pairs of physical
error faces, then deterministically
\(
|B|\le C X(E),
\)
while under iid noise
\(
X(E)=O(Np^2)
\)
with failure probability $\exp(-\Omega(Np^2))$.  Hence
\(
|B|=O(Np^2)
\)
with the same exponential tail.

The same collision count controls how much better an optimum can be than the
physical error:
\[
\OPT(S(E))
\ge
|E|-O(Np^2).
\]
Thus peeling leaves only a lower-order ambiguous part and, simultaneously, the
physical error is already within the same lower-order additive scale of the
global optimum.

The residual syndrome is
\[
S_{\mathrm{res}}
=
S(E)\triangle S(A)
=
S(B),
\]
so
\(
|S_{\mathrm{res}}|
\le
3|B|
=
O(Np^2).
\)
Since $s=\Theta(Np)$ with high probability whenever $Np\to\infty$, this may be
rewritten as
\(
|S_{\mathrm{res}}|=O(sp).
\)
Thus singleton peeling reduces a syndrome of typical size $\Theta(Np)$ to an
ambiguous residual smaller by another factor of $p$.

The residual ambiguity is still sparse enough to decompose into small local
components.  The key improvement over the sparse-regime analysis is that every
nontrivial residual component is caused by at least two nearby physical errors,
rather than by a single error.  This reduces the effective density governing
the growth of large residual components.  With the choice
\[
r=\Theta\left(\frac{\log N}{\log(1/p)}\right),
\]
one can show that every residual component has an exact witness of size
$o(\log N)$ with high probability whenever
\[
\frac{p\log N}{\log(1/p)}\to0.
\]
Each component can therefore be solved exactly by bounded search.  The resulting
correction satisfies
\(
|C|
\le
|E|+O(Np^2),
\)
while
\(
\OPT(S(E))
\ge
|E|-O(Np^2).
\)
Since $|E|=\Theta(Np)$ with high probability, this gives
\(
|C|
\le
(1+O(p))\OPT(S(E)).
\)

Because only the residual syndrome components are searched, the successful
random-instance running time is
\(
s^{1+o(1)}
\)
under the transition hypotheses.  As in the sparse regime, a capped search and
fallback give deterministic $N^{1+o(1)}$ worst-case time.

\paragraph{Dense regime: additive error becomes multiplicatively negligible.}
At still larger $p$, the useful mechanism changes again.  A periodic block
decomposition gives the deterministic additive guarantee
\[
\SOL\le \OPT+O(N/\tau)
\]
in time
\(
N2^{O(\tau)}.
\)
If every vertex is incident to exactly $d=O(1)$ faces, define
\[
\alpha(p)
:=
\frac{1-(1-2p)^d}{2}.
\]
This is the marginal probability that a vertex belongs to the syndrome.  A
bounded-dependency concentration argument gives
\[
|S|
\ge
\frac{N\alpha(p)}2
\]
with failure probability
\(
\exp(-\Omega(N\alpha(p))).
\)
Since
\(
\OPT(S)\ge |S|/3,
\)
the additive block-boundary loss becomes
\[
\frac{\SOL}{\OPT}
\le
1+
O\left(\frac1{\alpha(p)\tau}\right).
\]
Choosing
\[
\tau
=
\Theta\left(\frac1{\epsilon\alpha(p)}\right)
\]
gives a $(1+\epsilon)$-approximation in deterministic time
\(
N2^{O(1/(\epsilon\alpha(p)))}.
\)

On the same syndrome-density event,
\[
N=O\left(\frac{s}{\alpha(p)}\right),
\]
so the running time can equivalently be bounded by
\[
O\left(
\frac{s}{\alpha(p)}
2^{O(1/(\epsilon\alpha(p)))}
\right).
\]
For constant $p$, this simplifies to
\(
s2^{O(1/\epsilon)}.
\)
If
\(
\alpha(p)\log N\to\infty,
\)
and in particular if $p\log N\to\infty$ in the small-$p$ range, one can choose
$\tau=o(\log N)$ with $\alpha(p)\tau\to\infty$.  The resulting approximation is
$1+o(1)$ and the high-probability running time is
\(
s^{1+o(1)}.
\)

\subsection{Related Work}
\label{sec:related-work}

Minimum-weight decoding is a central optimization problem in quantum error
correction: given an observed syndrome, one seeks a minimum-weight error or
correction consistent with that syndrome.  Under iid code-capacity noise with
uniform error probabilities, minimum weight is a natural proxy for maximum
likelihood.  For the surface and toric codes, important instances reduce to
minimum-weight perfect matching, which has led to both strong theoretical
guarantees and highly optimized practical decoders
\cite{Edmonds:1965,dennis2001topologicalquantummemory}.  Color-code decoding is
more complicated because a single bulk error can create three syndrome defects
rather than a pair.  Several successful approaches circumvent this difficulty
by reducing the color code to simpler decoding problems, including projection
onto surface codes \cite{PhysRevA.89.012317}, the restriction decoder
\cite{Kubica2023efficientcolorcode}, matching on a M\"obius strip
\cite{PRXQuantum.3.010310}, and more recent concatenated matching constructions
\cite{Lee2025colorcodedecoder}.  Other families of topological-code decoders use
renormalization, belief propagation, or cluster growth
\cite{PhysRevLett.104.050504,Delfosse2021almostlineartime}.  These works
primarily study logical decoding performance and thresholds, whereas our focus
is the computational complexity of finding, or approximating, a minimum-weight
syndrome realization.

The worst-case complexity of this optimization problem was resolved only
recently.  Walters and Turner proved that minimum-weight decoding of the
two-dimensional color code is NP-hard
\cite{2603.04234}, and Gu, Wang, and Kubica independently established
NP-hardness for the color code together with several related topological
decoding settings \cite{2603.22064}.  This was followed by two independent
polynomial-time approximation schemes.  Walters gave a PTAS specifically for
planar color-code decoding
\cite{walters2026approximatelydecodingcolourcode}; under our convention that
the two-dimensional lattice has size \(N=\Theta(L^2)\), the explicit
running-time bound of the presented algorithm is
\(O(N^{2+528/\epsilon})\).  Gu, Wang, and Kubica gave a more general PTAS for
topological codes with point-like excitations connected by string-like
operators, with running time \(N(\log N)^{O(1/\epsilon)}\)
\cite{gu2026polynomialtimeapproximationschememinimumweight}.  Our results
address a complementary question: rather than approximating every syndrome, we
consider syndromes generated by a natural iid noise process.  This allows
stronger guarantees.  We obtain \(N^{1+o(1)}\) deterministic running time
throughout the noise range considered, with exact minimum-weight decoding when
\(p=o(1/\log^2 N)\) and a \(1+o(1)\) approximation above this scale.  Moreover,
when the syndrome is supplied sparsely, the running time on the corresponding
high-probability events is \((1+s)N^{o(1)}\) in the exact regime and
\(s^{1+o(1)}\) above it, where \(s\) is the observed syndrome size.

Our block decoder for the higher-density regime is related to a long line of
geometric and planar approximation schemes based on decomposing an instance
along periodically or recursively placed separators.  The shifting strategy of
Hochbaum and Maass obtains approximation schemes for geometric covering and
packing problems by selecting one of several shifted decompositions.
Baker's technique similarly removes periodically chosen layers of a planar
graph so that the remaining bounded-width pieces can be solved exactly by
dynamic programming \cite{1360292621319589888}.  Arora's PTAS for Euclidean
traveling salesman and related geometric problems instead uses a recursive
dissection together with a controlled set of boundary portals
\cite{b5565247f90d40daad3bf004914a571e}.  Both recent worst-case PTAS results
for color-code decoding draw directly on this geometric-approximation tradition,
and in particular on Arora's ideas
\cite{walters2026approximatelydecodingcolourcode,
gu2026polynomialtimeapproximationschememinimumweight}.  Our use of decomposition
is simpler and serves a different purpose.  We place periodic walls, solve the
resulting bounded-width regions exactly, and repair the syndrome remaining on
their boundaries.  This first gives a deterministic additive loss
\(O(N/\tau)\).  Randomness enters only afterward: concentration of the syndrome
size implies that the optimum is typically large enough for this additive
boundary cost to become a vanishing multiplicative error.  Thus, unlike a
worst-case PTAS, we do not need to modify an arbitrary optimum into a carefully
portal-respecting solution.

The peeling step used in the intermediate regime also has precedents, although
the way it is used here is different.  Peeling algorithms are classical tools
for decoding sparse-graph codes over the erasure channel.
Analogous procedures have been used for quantum erasure decoding, including
linear-time surface-code decoders
\cite{PhysRevResearch.2.033042}.  Cluster-growth decoders for topological codes
likewise often finish by peeling a locally constructed cluster
\cite{Delfosse2021almostlineartime}, and recent work on quantum LDPC erasure
decoding combines peeling with progressively more expensive processing of the
residual stopping set
\cite{yao2024clusterdecompositionimprovederasure}.  The transition decoder in
this paper does not have access to erasure locations.  Instead, it peels
syndrome-determined configurations that would be unambiguous for an isolated
physical error.  The specific analysis used here appears to be new to this
setting: every discrepancy between the peeled set and the unknown physical
error is charged to a nearby pair of physical errors.  Since such collisions
occur on the \(p^2\), rather than \(p\), scale, the unresolved instance has size
\(O(Np^2)\) with high probability.  We then solve only this residual ambiguity
exactly.  Thus the general peel-and-process-the-core paradigm is well
established, while the syndrome-only peeling rule and the pair-charging argument
are specific to our average-case analysis.

Finally, the sparse exact-decoding argument is closely connected to ideas from
percolation and the theory of lattice animals.  A lattice animal is a finite
connected collection of occupied sites or cells, and counting such connected
sets is a standard tool for controlling the sizes and diameters of random
clusters in subcritical percolation
\cite{grimmett:99:percolation,67ad6e18303b45078599c5be4f02cc83}.  More
generally, exponential bounds on the number of connected vertex sets in
bounded-degree graphs underlie many such union-bound arguments
\cite{43610bfd79984e30b86b41a88e91a967}.  Related locality principles also
appear algorithmically in renormalization-group and cluster-based decoders,
which group nearby syndrome defects and solve them locally or hierarchically
\cite{PhysRevLett.104.050504,Delfosse2021almostlineartime}.
Our sparse analysis uses a coarse analogue of a lattice-animal argument:
connected error sets in a graph of degree \(O(r^2)\) are exponentially
suppressed when the effective occupation parameter \(r^2p\) is small.  The
additional step needed for exact minimum-weight decoding is geometric rather
than purely percolative.  Small physical error clusters alone do not imply that
a globally optimal correction is local.  We combine the random-cluster bounds
with annulus-crossing and multi-region-crossing arguments showing that any
correction that escapes an inflated region or joins separated regions must pay
more than the available local witness.  This exchange argument forces every
global optimum to decompose into the same local problems, which is what turns
the usual small-cluster picture into an exact-decoding theorem.

\section{Preliminaries}
\label{sec:prelim}

\subsection{Syndromes and minimum-size realization}
\label{subsec:prelim-syndromes}

Let $G=(V,F)$ be a triangular color-code lattice, with
\[
N:=|V|
\qquad\text{and}\qquad
|F|=\Theta(N).
\]
For a face set $E\subseteq F$, define its syndrome by
\[
S(E)
=
\left\{
v\in V:
|\{f\in E:v\in f\}|\equiv1\pmod2
\right\}.
\]
The syndrome map is linear with respect to symmetric difference:
\begin{equation}
\label{eq:syndrome-linearity}
S(E\triangle E')
=
S(E)\triangle S(E').
\end{equation}
A syndrome $S\subseteq V$ is \emph{valid} if
\(
S=S(E)
\)
for some face set $E\subseteq F$.  For a valid syndrome $S$, define
\begin{equation}
\label{eq:OPT-definition}
\OPT(S)
:=
\min\{|C|:C\subseteq F,\ S(C)=S\}.
\end{equation}
For a syndrome $X\subseteq V$, we call any face set
$W\subseteq F$ satisfying $S(W)=X$ a \emph{witness} for $X$.
We write
\(
s:=|S|
\)
for the observed syndrome size.

We assume throughout that every face is triangular and every vertex is
incident to at most
\(
d=O(1)
\)
faces.  In the regular case, every vertex is incident to exactly $d$ faces.

The random error model is iid face noise: each face belongs to the physical
error $E$ independently with probability
\(
p=p(N).
\)
Unless stated otherwise, all probability statements are over this random
choice of $E$.

\paragraph{Sparse representation of the syndrome.}
Some of our running-time bounds are parameterized by the observed syndrome size
$s$ rather than by the full lattice size $N$.  For such statements we assume
that the input syndrome is supplied as a list of its violated vertices, together
with the ambient lattice description or local-access data structure.  We assume
that, from a vertex or face, all incident lattice objects can be enumerated in
$O(1)$ time, and that membership in the current syndrome set can be tested in
expected $O(1)$ time using hashing, or in logarithmic time using a deterministic
dictionary.

This distinction is necessary.  If the syndrome is presented as an $N$-bit
vector, reading the input already requires $\Omega(N)$ time.  The
syndrome-sensitive bounds therefore describe the work required once the
nonzero syndrome locations are given sparsely.

\subsection{Basic deterministic inequalities}
\label{subsec:prelim-basic}

Because every selected face toggles exactly three syndrome vertices, every face
set $C$ satisfies
\(
|S(C)|\le3|C|.
\)
Consequently,
\begin{equation}
\label{eq:syndrome-lower-bound}
\OPT(S)\ge\frac{s}{3}.
\end{equation}
For a physical error set $E$, define
\(
k_v:=|\{f\in E:v\in f\}|.
\)
Since every face contains exactly three vertices,
\begin{equation}
\label{eq:sum-kv-prelim}
\sum_{v\in V}k_v=3|E|.
\end{equation}
On the other hand,
\begin{equation}
\label{eq:syndrome-kv-prelim}
|S(E)|
=
\sum_{v\in V}(k_v\bmod2).
\end{equation}
Thus the difference between $3|E|$ and $|S(E)|$ is caused entirely by local
collisions in which two or more physical error faces meet a common vertex.
This observation will be used explicitly in the singleton-peeling analysis of
Section~\ref{sec:transition}.

\subsection{Coarse block geometry}
\label{subsec:prelim-block-geometry}

For the sparse localization arguments, fix once and for all a
\emph{constant-scale block decomposition} $\cB$, that is, a partition of the
faces into connected blocks, each containing $O(1)$ faces and having
$O(1)$ diameter.  Let $H_{\cB}$ be the graph whose vertices are the blocks in
$\cB$, with two blocks adjacent whenever some face of one shares a vertex with some
face of the other.
The periodic wall decomposition used by the dense-regime block decoder has a
different, variable scale $\tau$ and will be defined separately in
Section~\ref{sec:dense}.

For a block set $K\subseteq\cB$ and an integer $r\ge0$, define
\[
\cN_r(K)
:=
\{B\in\cB:\dist_{H_{\cB}}(B,K)\le r\}.
\]
Let $F_r(K)$ be the set of faces supported inside $\cN_r(K)$, and let
$V_r(K)$ be the set of vertices incident to those faces.  At constant block
scale, changing the precise convention at the boundary of $F_r(K)$ affects
all distances and cardinalities only by absolute constant factors.

We make the following geometric assumptions.

\begin{assumption}[Logarithmically two-dimensional geometry]
\label{ass:geometry}
There are absolute constants
\[
C_{\rm deg},C_{\rm gr},C_{\rm an},
\lambda_{\rm ann},\lambda_{\rm multi}>0
\]
and a scale
\[
L_N=\Omega(\log N)
\]
such that, for all integers $1\le r\le L_N$, the following hold.

\begin{enumerate}[label=(\roman*)]

\item \textbf{Bounded local degree.}
Every vertex is incident to at most $C_{\rm deg}$ faces, and the original
lattice, the face-adjacency graph, and the block graph $H_{\cB}$ have maximum
degree at most $C_{\rm deg}$.

\item \textbf{Quadratic growth.}
Every radius-$r$ block neighborhood contains at most
\(
C_{\rm gr}r^2
\)
blocks.

\item \textbf{Connected-animal counting.}
The number of connected face sets of cardinality $m$ containing any prescribed
face is at most
\(
C_{\rm an}^m.
\)

\item \textbf{Annulus crossing.}
If a connected face set $X$ meets $F_r(K)$ and also exits $F_{2r}(K)$, then
\(
|X|\ge\lambda_{\rm ann}r.
\)

\item \textbf{Multi-region crossing.}
Let $K_1,\ldots,K_\ell$ be block sets satisfying
\[
\dist_{H_{\cB}}(K_i,K_j)>6r
\qquad
(i\neq j).
\]
If a connected face set $X$ meets every region $F_r(K_i)$, then
\(
|X|
\ge
\lambda_{\rm multi}r(\ell-1).
\)

\end{enumerate}
\end{assumption}

These are the only large-scale geometric properties used by the sparse
localization proof.  The standard planar triangular color-code lattice has
these properties. 

\begin{lemma}[Geometry of the triangular color-code lattice]
\label{lem:standard-geometry}
For the standard planar triangular color-code lattice, equipped with any fixed
constant-scale block decomposition of bounded block diameter,
Assumption~\ref{ass:geometry} holds.
\end{lemma}

\begin{proof}
The lattice and its face-adjacency graph have uniformly bounded degree.
The same is true of the block graph because every block contains only
$O(1)$ faces and has bounded diameter.  This proves~(i).

The block graph is a bounded-scale discretization of a two-dimensional
triangular lattice.  Hence a radius-$r$ block ball contains $O(r^2)$ blocks,
which proves~(ii).

Property~(iii) follows from bounded degree: in any graph of maximum degree
$\Delta=O(1)$, the number of connected vertex sets of cardinality $m$
containing a fixed vertex is at most $C^m$ for a constant $C$ depending only
on $\Delta$.

For~(iv), suppose that a connected face set $X$ meets $F_r(K)$ and exits
$F_{2r}(K)$.  A face-adjacency path in $X$ from the first point to the second
must cross block distance $\Omega(r)$.  Since one face-adjacency step changes
block distance by at most a constant, such a path contains $\Omega(r)$ faces.
Thus
\(
|X|\ge \lambda_{\rm ann}r
\)
for some lattice-dependent constant $\lambda_{\rm ann}>0$.

For~(v), choose for each $i$ a face
\(
f_i\in X\cap F_r(K_i).
\)
Since
\(
\dist_{H_{\cB}}(K_i,K_j)>6r,
\)
the bounded block diameter implies that any two chosen faces $f_i,f_j$ are at
face-adjacency distance $\Omega(r)$.  Let $T$ be a spanning tree of the
connected face set $X$.  A depth-first traversal of $T$ has length
$2|E(T)|$ and visits all $\ell$ chosen faces.  Shortcutting this traversal in
the face-adjacency metric gives a closed walk through the $f_i$'s.  Since
distinct $f_i$'s are pairwise at distance $\Omega(r)$, this walk has length
$\Omega(r\ell)$.  Consequently
\(
|X|
\ge
\lambda_{\rm multi}r(\ell-1)
\)
for some constant $\lambda_{\rm multi}>0$.

All of these bounds hold at least up to logarithmic distance from the lattice
boundary, so one may take $L_N=\Omega(\log N)$.
\end{proof}

Set
\( \lambda_0
:=
\min\{\lambda_{\rm ann},\lambda_{\rm multi}\}\).
Because the block decomposition has constant scale, every syndrome vertex is
incident to only $O(1)$ blocks.  Hence a set of $s$ syndrome vertices can
activate at most
\( O(s) \)
blocks.  This elementary fact will be used in the syndrome-sensitive running
time of the sparse decoder.

We will also use the following elementary consequence of bounded local degree.

\begin{lemma}[Syndrome vertices generate only linearly many local candidates]
\label{lem:syndrome-local-candidates}
Let $X\subseteq V$ be any vertex set.  The number of faces incident to at least
one vertex of $X$ is at most
\(
C_{\rm deg}|X|.
\)
Consequently, for every fixed radius $R=O(1)$, the radius-$R$ neighborhood in
the face-adjacency graph of the faces incident to $X$ contains
\(
O(|X|)
\)
faces, where the hidden constant may depend on $R$ and the lattice constants.
\end{lemma}

\begin{proof}
By Assumption~\ref{ass:geometry}, each vertex is incident to at most
$C_{\rm deg}$ faces.  Hence at most
\(
C_{\rm deg}|X|
\)
faces are incident to vertices of $X$.

The face-adjacency graph also has maximum degree at most $C_{\rm deg}$.
Therefore, for every fixed $R$, the radius-$R$ neighborhood of any face
contains only $O(1)$ faces.  Taking the union over the at most
$C_{\rm deg}|X|$ incident faces gives the claimed $O(|X|)$ bound.
\end{proof}

This lemma is what allows the sparse and transition decoders to discover their
relevant local structures from a sparse syndrome without first scanning the
entire lattice.
\subsection{Bounded-search exact solver}
\label{subsec:prelim-bounded-search}

The sparse and transition algorithms repeatedly solve small exact syndrome
realization problems.  Their basic parameterized primitive is the following.

\begin{lemma}[Bounded-search solver]
\label{lem:bounded-search}
Let $X$ be a valid target syndrome and let $k\ge0$.  One can decide whether
there is a correction $C$ satisfying
\[
S(C)=X,
\qquad
|C|\le k,
\]
and, if so, return such a correction, in time
\(
d^{O(k)}\poly(k+|X|).
\)
By iterative deepening, if
\(
\OPT(X)\le K,
\)
one can find an exact optimum in time
\(
d^{O(K)}\poly(K+|X|).
\)

The same statement holds when the correction is required to lie inside a
specified local face region $Q$, provided branching is restricted to faces of
$Q$.
\end{lemma}

\begin{proof}
If $X=\varnothing$, the empty correction is feasible.  Otherwise choose any
violated vertex
\(
v\in X.
\)
Every correction $C$ satisfying $S(C)=X$ must contain at least one face
incident to $v$: otherwise the parity at $v$ would remain zero.

There are at most $d$ faces incident to $v$.  Branch on the choice of one such
face $f$.  After selecting $f$, replace
\(
X
\leftarrow
X\triangle S(\{f\})
\)
and reduce the remaining budget by one.

To verify completeness, suppose $C^\star$ is a feasible correction with
$|C^\star|\le k$.  Since $v\in X=S(C^\star)$, at least one face
$f\in C^\star$ is incident to $v$.  In the branch choosing this $f$, the set
\(
C^\star\setminus\{f\}
\)
realizes the residual target
\(
X\triangle S(\{f\})
\)
with budget at most $k-1$.  Thus at least one branch preserves a feasible
solution.

The search tree has branching factor at most $d$ and depth at most $k$, so it
has $d^{O(k)}$ nodes.  Maintaining the current syndrome as a sparse set gives
only polynomial overhead in $k+|X|$.

For a restricted local region $Q$, simply discard branches selecting a face
outside $Q$.  The same completeness argument applies to any feasible solution
contained in $Q$.

Finally, iterative deepening over
\(
k=0,1,\ldots,K
\)
returns the smallest feasible budget and therefore an exact optimum.  The
geometric sum of the search-tree sizes is still $d^{O(K)}$.
\end{proof}

\begin{remark}[What the bounded-search time depends on]
\label{rem:bounded-search-local}
The exponential part of the running time depends on the optimum weight bound
$K$, not on the number of faces in the ambient lattice.  This is crucial for
the syndrome-sensitive algorithms: once a small residual syndrome component
has been isolated, exact search explores only the faces reached by its bounded
branching process.
\end{remark}

\subsection{Running-time conventions}
\label{subsec:prelim-running-time}

We use two related running-time notions.

First, the sparse and transition algorithms are equipped with deterministic
caps.  If a local structural condition fails or a bounded search exceeds its
prescribed budget, the algorithm invokes the fixed-spacing block decoder from
Section~\ref{sec:dense}.  This gives a deterministic worst-case running time
of
\(
N^{1+o(1)}
\)
for the complete algorithms.

Second, we distinguish worst-case running time from typical running time under
the iid noise model.  The capped algorithms run in $N^{1+o(1)}$ time on every
input.  With high probability over the random error, no fallback is needed.
For a sparse syndrome representation, the algorithms then inspect only the
local neighborhoods generated by the observed syndrome, yielding sharper
running times in terms of $s=|S|$.

For the dense-regime block decoder, the algorithm processes the entire lattice,
so its deterministic complexity is naturally expressed in terms of $N$.
Under iid noise, however, syndrome concentration gives
\[
N=O\left(\frac{s}{\alpha(p)}\right)
\]
with high probability, allowing the same running time to be re-expressed in
terms of $s$.

\subsection{Organization of the proofs}
\label{subsec:prelim-organization}

The proof is divided into three algorithmic regimes and then combined.

\begin{itemize}

\item Section~\ref{sec:dense} proves the deterministic block-decoder additive
guarantee, syndrome concentration, the general approximation tradeoff, and its
syndrome-sensitive consequences.

\item Section~\ref{sec:transition} develops singleton peeling, charges peeling
errors to nearby physical-error pairs, analyzes the residual polymer
components, and proves the $(1+O(p))$ transition theorem.  The same analysis
shows that the residual syndrome has size $O(sp)$ and can be solved in
$s^{1+o(1)}$ time with high-probability.

\item Section~\ref{sec:sparse} proves that, below the $1/\log^2N$ scale, the random
syndrome decomposes with high probability into separated local regions over
which every global optimum decomposes.  This yields exact decoding with high
probability.  If the syndrome is supplied sparsely, then on the same
high-probability event the running time is
\(
(1+s)N^{o(1)}.
\)

\item Section~\ref{sec:combined} combines the three algorithms into the
two-regime guarantee given in Table \ref{table:two_regime}.

\end{itemize}

\section{Dense Regime: Block Approximation}
\label{sec:dense}

This section analyzes a deterministic block decoder and then converts its
additive error guarantee into a multiplicative approximation guarantee under
iid face noise.  The only probabilistic input is a lower-tail estimate for the
syndrome size.  The decoder itself processes the entire lattice and therefore
has a deterministic running time naturally expressed in terms of $N$.  On the
same high-probability syndrome-density event used for the approximation
analysis, however, $N$ can be bounded in terms of the observed syndrome size
\(
s:=|S|,
\)
which yields an equivalent syndrome-sensitive running-time bound.

\subsection{The block decoder}
\label{subsec:dense-block-decoder}

Fix an integer wall spacing $\tau\ge 1$.  We use the standard periodic wall
decomposition of the triangular lattice.  Concretely, choose periodically
spaced vertical and anti-diagonal columns of triangular faces, with consecutive
walls separated by $\Theta(\tau)$ lattice spacings, and let
\(
F_{\partial}\subseteq F
\)
be the union of the selected wall faces.  A vertex is a \emph{boundary vertex}
if it is incident to a boundary face:
\[
V_{\partial}
:=
\{v\in V:\text{$v$ is incident to some }f\in F_{\partial}\}.
\]
Set
\[
F_{\mathrm{int}}:=F\setminus F_{\partial},
\qquad
V_{\mathrm{int}}:=V\setminus V_{\partial}.
\]

A \emph{block} is a maximal connected component of the sublattice induced by
$F_{\mathrm{int}}$, where two faces are adjacent when they share a vertex.
Write
\(
\mathcal B=\{B_1,\ldots,B_q\}
\)
for the collection of blocks.  For a block $B_i$, let $F(B_i)$ be its face set,
let
\[
V(B_i):=\{v\in V:\text{$v$ is incident to a face in }F(B_i)\},
\]
and define its interior vertices by
\(
I(B_i):=V(B_i)\cap V_{\mathrm{int}}.
\)

The periodic construction has the following elementary geometric properties:
\begin{equation}
\label{eq:wall-counts}
q=O(N/\tau^2),
\qquad
|F(B_i)|=O(\tau^2),
\qquad
|F_{\partial}|=O(N/\tau).
\end{equation}
Moreover, each block admits a sweep ordering of its faces with frontier width
$O(\tau)$.  Indeed, sweep the faces one lattice row at a time.  At any stage,
vertices incident to both processed and unprocessed faces lie within a
constant-width neighborhood of the current row, which contains $O(\tau)$
vertices. This is the width property used by the local dynamic program
below.

Given a valid input syndrome $S$, define the local target in block $B_i$ by
\(
S_i:=S\cap I(B_i).
\)
The decoder solves the local optimization problem
\begin{equation}
\label{eq:dense-local-problem}
\min\left\{
|E_i|:
E_i\subseteq F(B_i),\
S(E_i)\cap I(B_i)=S_i
\right\}.
\end{equation}
There is deliberately no parity constraint on the boundary vertices.  Let
$E_i$ be an optimum solution of~\eqref{eq:dense-local-problem}, and define
\(
E_{\mathrm{loc}}
:=
\bigtriangleup_{i=1}^q E_i.
\)
Since the blocks have disjoint face sets, this symmetric difference is simply
their union.

\paragraph{Exact solution of a block.}
We briefly record why~\eqref{eq:dense-local-problem} can be solved in
$\tau^2 2^{O(\tau)}$ time.  Fix a sweep ordering
\(
f_1,f_2,\ldots,f_m
\)
of $F(B_i)$, where $m=O(\tau^2)$.  After processing
\(
P_k:=\{f_1,\ldots,f_k\},
\)
let $W_k\subseteq I(B_i)$ be the set of interior vertices incident to both a
processed and an unprocessed face.  By the sweep-width property stated above,
\(
|W_k|=O(\tau)
\)
for every $k$.

A dynamic-programming state is the parity vector
\(
\sigma\in\{0,1\}^{W_k}.
\)
Let $D_k(\sigma)$ be the minimum number of selected faces among $P_k$ such
that
\begin{enumerate}[label=(\roman*)]
\item the parity currently induced on every $v\in W_k$ is $\sigma(v)$; and
\item every interior vertex all of whose incident block faces have already
been processed has its required final parity
$\mathbf 1[v\in S_i]$.
\end{enumerate}
When processing $f_k$, the transition branches on whether $f_k$ is selected,
updates the parities of its incident interior vertices, and rejects a transition
if a vertex leaving the frontier has the wrong target parity. Vertices in $V(B_i)\cap V_{\partial}$ are not stored in the frontier state,
because the local problem imposes no parity constraint on them.

There are $2^{O(\tau)}$ states at each of $O(\tau^2)$ sweep positions and two
transitions per state.  Thus an exact local optimum, together with a realizing
face set, can be recovered in
\(
\tau^2 2^{O(\tau)}
\)
time.

To recover an optimal face set, whenever a transition attains the minimum
defining $D_k(\sigma')$, we store its predecessor state together with the
choice $x\in\{0,1\}$ indicating whether $f_k$ was selected.  After all
$m$ faces have been processed, the frontier is empty, and backtracking from
$D_m(\emptyset)$ recovers an optimal set
\(
E_i\subseteq F(B_i).
\)

After all blocks have been solved, define the residual syndrome
\begin{equation}
\label{eq:dense-residual}
R:=S\triangle S(E_{\mathrm{loc}}).
\end{equation}
For every interior vertex $v\in I(B_i)$, the local optimization enforces the
correct parity at $v$.  Hence
\(
R\subseteq V_{\partial}.
\)
It remains to repair this boundary-supported syndrome.

\paragraph{Boundary graphs.}
Let
\(
V=V_R\sqcup V_G\sqcup V_B
\)
be the proper three-coloring of the vertices.  For
$c\in\{R,G,B\}$ and any syndrome $X$, write
\(
X_c:=X\cap V_c.
\)
Every triangular face contains exactly one vertex of each color.

For each color $c$, define a graph
\(
H_c=(V_{\partial}\cap V_c,A_c),
\)
where $A_c$ is the following set of edges.  Two vertices
$u,v\in V_{\partial}\cap V_c$ are joined by an edge $a$ if there exist
two adjacent boundary faces
\(
M_a=\{f_a^1,f_a^2\}\subseteq F_{\partial}
\)
such that
\(
S(M_a)=\{u,v\}.
\)
Indeed, two adjacent triangular faces share an edge containing vertices of the
other two colors, so their two nonshared vertices have the same color.  For
each edge $a\in A_c$, fix one such pair of boundary faces $M_a$.

The periodic walls are chosen so that the following three properties hold:
\begin{enumerate}[label=(B\arabic*)]
\item each $H_c$ is connected;
\item there is a boundary face $f^\star$ whose three vertices
$x_R^\star,x_G^\star,x_B^\star$ belong respectively to
$H_R,H_G,H_B$; and
\item
\begin{equation}
|A_R|+|A_G|+|A_B|
=
O(N/\tau).
\label{eq:boundary_graph_size}
\end{equation}
\end{enumerate}
These properties follow directly from the periodic wall geometry.  The walls
have total length $O(N/\tau)$, and the pairs of adjacent boundary faces
associated with the edges of $H_c$ connect the boundary vertices of each color
along the wall network.

We first recall the elementary parity fact.

\begin{lemma}[Tree parity correction]
\label{lem:dense-tree-parity}
Let $T=(U,A)$ be a tree and let $X\subseteq U$ have even cardinality.  There
is a set of edges $A_X\subseteq A$ whose odd-degree vertices are exactly $X$.
\end{lemma}

\begin{proof}
Root $T$ at an arbitrary vertex.  For every nonroot vertex $u$, let $T_u$ be
the subtree rooted at $u$, and let $e_u$ be the edge from $u$ to its parent.
Set
\[
A_X
:=
\{e_u:|X\cap V(T_u)|\equiv 1\pmod 2\}.
\]
For a nonroot vertex $u$ with children $u_1,\ldots,u_k$, the incident tree
edges consist of the parent edge of $u$ and the child edges $uu_1,\ldots,uu_k$.
By the definition of $A_X$,
\[
\deg_{A_X}(u)
\equiv
|X\cap V(T_u)|
+
\sum_{j=1}^k |X\cap V(T_{u_j})|
\pmod 2.
\]
Since
\(
V(T_u)
=
\{u\}\sqcup
V(T_{u_1})\sqcup\cdots\sqcup V(T_{u_k}),
\)
we have
\[
|X\cap V(T_u)|
=
\mathbf 1[u\in X]
+
\sum_{j=1}^k |X\cap V(T_{u_j})|.
\]
Therefore
\begin{align*}
\deg_{A_X}(u)
&\equiv
\mathbf 1[u\in X]
+
2\sum_{j=1}^k |X\cap V(T_{u_j})|
\pmod 2\\
&\equiv
\mathbf 1[u\in X]
\pmod 2.
\end{align*}
For the root $r$, there is no parent edge, so
\[
\deg_{A_X}(r)
\equiv
\sum_{j=1}^k |X\cap V(T_{u_j})|
\pmod 2,
\]
where $u_1,\ldots,u_k$ are the children of $r$.  Since
\(
V(T)
=
\{r\}\sqcup
V(T_{u_1})\sqcup\cdots\sqcup V(T_{u_k}),
\)
we have
\[
|X|
=
\mathbf 1[r\in X]
+
\sum_{j=1}^k |X\cap V(T_{u_j})|.
\]
Because $|X|$ is even,
\[
\sum_{j=1}^k |X\cap V(T_{u_j})|
\equiv
\mathbf 1[r\in X]
\pmod 2.
\]
Therefore
\(
\deg_{A_X}(r)
\equiv
\mathbf 1[r\in X]
\pmod 2.
\)
Hence, for every $u\in U$,
\(
\deg_{A_X}(u)
\equiv
\mathbf 1[u\in X]
\pmod 2.
\)
Therefore, precisely the vertices in $X$ have odd degree in $(U,A_X)$.
\end{proof}

\begin{lemma}[Color-parity invariant]
\label{lem:dense-color-parity}
For every face set $E'\subseteq F$,
\[
|S(E')\cap V_R|
\equiv
|S(E')\cap V_G|
\equiv
|S(E')\cap V_B|
\pmod 2.
\]
Consequently,
\(
|R_R|\equiv |R_G|\equiv |R_B|\pmod 2.
\)
\end{lemma}

\begin{proof}
Toggling one triangular face changes the syndrome parity of exactly one red,
one green, and one blue vertex.  Starting from the empty face set, the three
color-class parities therefore remain equal.  Since
\[
R=S\triangle S(E_{\mathrm{loc}})
=
S(E^\prime)
\]
for any face set $E^\prime$ realizing this symmetric difference, the same
invariant holds for $R$.  Equivalently, it follows directly from linearity of
the syndrome map.
\end{proof}

\begin{lemma}[Boundary repair]
\label{lem:dense-boundary-repair}
For every residual syndrome $R\subseteq V_{\partial}$ satisfying the
color-parity invariant, one can construct
$E_{\partial}\subseteq F_{\partial}$ such that
\(
S(E_{\partial})=R
\)
and
\(
|E_{\partial}|=O(N/\tau).
\)
The construction takes $O(N/\tau)$ time after the boundary graphs and spanning
trees have been fixed.
\end{lemma}

\begin{proof}
Fix a spanning tree $T_c$ of each $H_c$.

Suppose first that $|R_R|,|R_G|,|R_B|$ are all even. Apply
Lemma~\ref{lem:dense-tree-parity} to $R_c$ in $T_c$. For each selected
tree edge $a$, include the fixed pair of boundary faces associated with $a$.
By construction, the syndrome of this pair is exactly the set of endpoints
of $a$. Therefore, taking the symmetric difference over all selected tree
edges gives syndrome exactly $R_c$. Performing this construction
independently for the three colors gives total syndrome $R$.

The number of selected tree edges is at most the total number of tree edges,
and each selected edge contributes at most two boundary faces. Thus,
by~\eqref{eq:boundary_graph_size}, the total number of boundary faces used is
$O(N/\tau)$.

Now suppose the three cardinalities are odd. Choose one vertex
$x_c\in R_c$ for each color. First repair the even set
\(
R_c\setminus\{x_c\}
\)
as above. Next apply Lemma~\ref{lem:dense-tree-parity} to the two-element set
\(
\{x_c,x_c^\star\}
\)
in $T_c$. Including the fixed pairs of boundary faces associated with the
selected tree edges changes the syndrome by exactly
\(
\{x_c,x_c^\star\}.
\)
Hence the remaining color-$c$ syndrome is transferred from $x_c$ to the
distinguished vertex $x_c^\star$.

After carrying this out for all three colors, the only remaining syndrome is
\(
\{x_R^\star,x_G^\star,x_B^\star\}
=
S(\{f^\star\}).
\)
Including the boundary face $f^\star$ cancels this syndrome. The total number
of selected tree edges remains $O(N/\tau)$, and therefore the total number of
boundary faces used is also $O(N/\tau)$.

Finally, a single bottom-up traversal of each spanning tree identifies the
required parity subgraph in time linear in the tree size, which proves the
running-time claim.
\end{proof}

The \emph{block decoder} outputs
\(
E_{\mathrm{blk}}
:=
E_{\mathrm{loc}}\triangle E_{\partial}.
\)

\subsection{Deterministic additive guarantee}
\label{subsec:dense-additive}

The key deterministic fact is that the exact local solutions never cost more
than the restriction of a global optimum; the only possible loss comes from
repairing the boundary.

\begin{lemma}[Deterministic additive guarantee]
\label{lem:dense-additive}
For every valid syndrome $S$, the block decoder outputs a correction satisfying
\(
S(E_{\mathrm{blk}})=S
\)
and
\begin{equation}
\label{eq:dense-additive}
|E_{\mathrm{blk}}|
\le
\OPT(S)+C_{\partial}\frac{N}{\tau}
\end{equation}
for a lattice-dependent constant $C_{\partial}>0$.
\end{lemma}

\begin{proof}
Let $E^\star$ be a minimum-size realization of $S$:
\[
S(E^\star)=S,
\qquad
|E^\star|=\OPT(S).
\]
Fix a block $B_i$.  An interior vertex in $I(B_i)$ is not incident to any
boundary face.  Moreover, all interior faces incident to that vertex belong to
the same block.  Therefore
\[
E_i^\star:=E^\star\cap F(B_i)
\]
has exactly the required syndrome on the interior vertices:
\(
S(E_i^\star)\cap I(B_i)
=
S\cap I(B_i)
=
S_i.
\)
Thus $E_i^\star$ is feasible for the local problem
\eqref{eq:dense-local-problem}.  Since $E_i$ is a minimum local solution,
\(
|E_i|\le |E_i^\star|.
\)
The block face sets are disjoint, so
\[
|E_{\mathrm{loc}}|
=
\sum_i |E_i|
\le
\sum_i |E^\star\cap F(B_i)|
\le
|E^\star|
=
\OPT(S).
\]

By construction, the residual $R$ in~\eqref{eq:dense-residual} is supported on
$V_{\partial}$.  Lemma~\ref{lem:dense-boundary-repair} gives
$E_{\partial}\subseteq F_{\partial}$ satisfying
\[
S(E_{\partial})=R,
\qquad
|E_{\partial}|\le C_{\partial}\frac{N}{\tau}.
\]
Consequently,
\[
S(E_{\mathrm{blk}})
=
S(E_{\mathrm{loc}})\triangle S(E_{\partial})
=
S(E_{\mathrm{loc}})\triangle R
=
S.
\]
Finally,
\[
|E_{\mathrm{blk}}|
\le
|E_{\mathrm{loc}}|+|E_{\partial}|
\le
\OPT(S)+C_{\partial}\frac{N}{\tau}. \hfill \qedhere
\]
\end{proof}

\begin{lemma}[Deterministic running time]
\label{lem:dense-runtime}
For every input syndrome and every fixed wall spacing $\tau$, the block decoder
runs in
\(
N2^{O(\tau)}
\)
time.
\end{lemma}

\begin{proof}
By~\eqref{eq:wall-counts}, there are $O(N/\tau^2)$ blocks, and each local problem costs
$\tau^2 2^{O(\tau)}$.  Thus all local dynamic programs together cost
\[
O(N/\tau^2)\cdot \tau^2 2^{O(\tau)}
=
N2^{O(\tau)}.
\]
The boundary repair is $O(N/\tau)$ and is therefore lower order.
\end{proof}

\subsection{Syndrome density under iid face noise}
\label{subsec:dense-syndrome}

We now assume iid face noise with rate $p$ and, for this section, that every
vertex is incident to exactly $d=O(1)$ faces.  Let
\(
S:=S(E).
\)
For a vertex $v$, write
\(
Y_v:=\mathbf 1[v\in S].
\)

\begin{lemma}[Syndrome marginal]
\label{lem:dense-alpha}
For every vertex $v$,
\[
\Pr(v\in S)
=
\alpha(p)
:=
\frac{1-(1-2p)^d}{2}.
\]
Consequently,
\(
\mu:=\mathbb E|S|=N\alpha(p).
\)
For $p=o(1)$,
\[
\alpha(p)=dp+O(p^2)=\Theta(p).
\]
\end{lemma}

\begin{proof}
The number of erroneous faces incident to $v$ is a
$\operatorname{Bin}(d,p)$ random variable.  If $Z$ has this distribution, then
\[
\mathbb E[(-1)^Z]
=
(1-2p)^d.
\]
On the other hand,
\[
\mathbb E[(-1)^Z]
=
\Pr(Z\text{ even})-\Pr(Z\text{ odd})
=
1-2\Pr(Z\text{ odd}).
\]
Solving for the odd-parity probability gives
\[
\Pr(v\in S)
=
\Pr(Z\text{ odd})
=
\frac{1-(1-2p)^d}{2}.
\]
Summing over all $N$ vertices proves the expectation formula.  The small-$p$
expansion follows from the binomial theorem:
\[
(1-2p)^d=1-2dp+O(p^2).
\]
\end{proof}

The variables $Y_v$ are not independent, but they have bounded dependency.

\begin{lemma}[Lower tail for the syndrome size]
\label{lem:dense-syndrome-tail}
Let
\(
q:=2d+1.
\)
Then
\begin{equation}
\label{eq:dense-tail}
\Pr\left(
|S|<\frac{\mu}{2}
\right)
\le
q\exp\left(
-\frac{\mu}{8q^2}
\right).
\end{equation}
Equivalently,
\[
\Pr\left(
|S|<\frac{N\alpha(p)}{2}
\right)
\le
(2d+1)
\exp\left(
-\frac{N\alpha(p)}{8(2d+1)^2}
\right).
\]
For constant $d$, the failure probability is
\(
\exp(-\Omega(N\alpha(p))).
\)
\end{lemma}

\begin{proof}
Construct a dependency graph $D$ on the lattice vertices by joining distinct
vertices $u,v$ whenever some triangular face is incident to both.  The variable
$Y_v$ depends only on the independent error indicators of the faces incident to
$v$.  Therefore two collections of variables indexed by sets with no edge
between them in $D$ depend on disjoint sets of face-error variables.

Every vertex belongs to at most $d$ triangular faces, and each such face
contains at most two other vertices.  Hence
\(
\Delta(D)\le 2d.
\)
Properly color $D$ with
\(
q=2d+1
\)
colors, and write
\[
V=V_1\sqcup\cdots\sqcup V_q.
\]
Within each color class the random variables $\{Y_v:v\in V_j\}$ are mutually
independent.  Define
\[
X_j:=\sum_{v\in V_j}Y_v,
\qquad
\mu_j:=\mathbb E [ X_j].
\]
Then
\[
|S|=\sum_{j=1}^q X_j,
\qquad
\mu=\sum_{j=1}^q\mu_j.
\]
If $|S|<\mu/2$, then
\[
\sum_{j=1}^q(\mu_j-X_j)>\frac{\mu}{2}.
\]
Hence for at least one $j$,
\[
X_j<\mu_j-\frac{\mu}{2q}.
\]
If $\mu_j\le\mu/(2q)$, then
\(
\mu_j-\frac{\mu}{2q}\le0.
\)
Since $X_j\ge0$, the event
\(
X_j<\mu_j-\frac{\mu}{2q}
\)
cannot occur.  Otherwise, the additive Chernoff bound for a sum of independent
Bernoulli variables gives
\[
\Pr\left(
X_j<\mu_j-\frac{\mu}{2q}
\right)
\le
\exp\left(
-\frac{\mu^2}{8q^2\mu_j}
\right).
\]
Since $\mu_j\le\mu$,
\[
\Pr\left(
X_j<\mu_j-\frac{\mu}{2q}
\right)
\le
\exp\left(
-\frac{\mu}{8q^2}
\right).
\]
A union bound over the $q$ color classes yields~\eqref{eq:dense-tail}.
\end{proof}

The same concentration event gives the relation between lattice size and
observed syndrome size that will be used below.

\begin{corollary}[Lattice size in terms of syndrome size]
\label{cor:dense-N-vs-s}
On the event
\[
|S|\ge \frac{N\alpha(p)}2,
\]
we have
\begin{equation}
\label{eq:dense-N-vs-s}
N\le \frac{2s}{\alpha(p)}.
\end{equation}
Consequently, any deterministic running-time bound of the form
\(
N\,g(N,p)
\)
for the block decoder can, on this event, be written as
\[
O\left(\frac{s}{\alpha(p)}\,g(N,p)\right).
\]
\end{corollary}

\begin{proof}
Since $s=|S|$, the assumed lower bound gives
\(
s\ge \frac{N\alpha(p)}2.
\)
Rearranging proves~\eqref{eq:dense-N-vs-s}.
\end{proof}

\subsection{General approximation tradeoff}
\label{subsec:dense-tradeoff}

We now combine the deterministic additive bound with syndrome concentration.

\begin{theorem}[Average-case approximation tradeoff]
\label{thm:dense-tradeoff}
Assume every vertex is incident to exactly $d=O(1)$ faces and define
\[
\alpha(p):=\frac{1-(1-2p)^d}{2}.
\]
For every wall spacing $\tau\ge1$, with probability at least
\begin{equation}
\label{eq:dense-tradeoff-prob}
1-
(2d+1)
\exp\left(
-\frac{N\alpha(p)}{8(2d+1)^2}
\right),
\end{equation}
the block decoder satisfies
\begin{equation}
\label{eq:dense-general-ratio}
\frac{|E_{\mathrm{blk}}|}{\OPT(S)}
\le
1+
\frac{C_{\mathrm{app}}}{\alpha(p)\tau},
\end{equation}
for a lattice-dependent constant $C_{\mathrm{app}}>0$.  Its running time is
deterministically
\begin{equation}
\label{eq:dense-runtime-general}
N2^{O(\tau)}.
\end{equation}
On the same event~\eqref{eq:dense-tradeoff-prob}, this running time obeys the
syndrome-sensitive bound
\begin{equation}
\label{eq:dense-runtime-general-s}
O\left(
\frac{s}{\alpha(p)}2^{O(\tau)}
\right).
\end{equation}
Consequently, for every $\epsilon>0$, choosing
\[
\tau
=
\left\lceil
\frac{C_{\mathrm{app}}}{\epsilon\alpha(p)}
\right\rceil
\]
gives
\[
|E_{\mathrm{blk}}|
\le
(1+\epsilon)\OPT(S)
\]
with probability~\eqref{eq:dense-tradeoff-prob}.  The deterministic running
time is
\begin{equation}
\label{eq:dense-epsilon-time}
N2^{O(1/(\epsilon\alpha(p)))},
\end{equation}
and on the same high-probability event it is bounded by
\begin{equation}
\label{eq:dense-epsilon-time-s}
O\left(
\frac{s}{\alpha(p)}
2^{O(1/(\epsilon\alpha(p)))}
\right).
\end{equation}
\end{theorem}

\begin{proof}
On the event
\(
|S|\ge \frac{N\alpha(p)}2,
\)
the deterministic bound
\(
\OPT(S)\ge \frac{|S|}{3}
\)
from~\eqref{eq:syndrome-lower-bound} gives
\(
\OPT(S)\ge \frac{N\alpha(p)}6.
\)
By Lemma~\ref{lem:dense-additive},
\[
|E_{\mathrm{blk}}|
\le
\OPT(S)+C_{\partial}\frac{N}{\tau}.
\]
If the syndrome is nonempty, division by $\OPT(S)$ yields
\begin{align*}
\frac{|E_{\mathrm{blk}}|}{\OPT(S)}
\le
1+
\frac{C_{\partial}N}{\tau\OPT(S)}
\le
1+
\frac{6C_{\partial}}{\alpha(p)\tau}.
\end{align*}
Thus~\eqref{eq:dense-general-ratio} holds with
\(
C_{\mathrm{app}}:=6C_{\partial}.
\)
Lemma~\ref{lem:dense-syndrome-tail} gives the probability of this event.
The deterministic runtime is Lemma~\ref{lem:dense-runtime}.  On the same event,
Corollary~\ref{cor:dense-N-vs-s} gives
\(
N\le \frac{2s}{ \alpha(p)},
\)
and substituting this into the deterministic runtime proves
\eqref{eq:dense-runtime-general-s}.

Finally, if
\(
\tau\ge \frac{C_{\mathrm{app}}}{\epsilon\alpha(p)},
\)
then the excess term in~\eqref{eq:dense-general-ratio} is at most $\epsilon$,
and substituting this value of $\tau$ into the running-time bounds gives
\eqref{eq:dense-epsilon-time} and
\eqref{eq:dense-epsilon-time-s}.
\end{proof}

\begin{remark}
The high-probability statement is meaningful whenever
\(
N\alpha(p)\to\infty.
\)
For $p=o(1)$ and fixed $d$, this condition is equivalent to
\(
Np\to\infty.
\)
\end{remark}

\subsection{Near-linear consequence when
\texorpdfstring{\(p\log N\to\infty\)}{p log N tends to infinity}}
\label{subsec:dense-nearlinear}

The general tradeoff immediately yields a near-linear-time
$(1+o(1))$ approximation once the syndrome density is asymptotically larger
than $1/\log N$.

\begin{corollary}[Near-linear near-optimal decoding]
\label{cor:dense-nearlinear}
Assume
\begin{equation}
\label{eq:alpha-log-condition}
\alpha(p)\log N\to\infty.
\end{equation}
Then the block decoder can be parameterized so that its deterministic running
time is
\(
N^{1+o(1)}
\)
and, with probability
\(
1-\exp(-\Omega(N\alpha(p))),
\)
its output satisfies
\(
|E_{\mathrm{blk}}|
=
(1+o(1))\OPT(S).
\)
On the same high-probability event, its running time is
\(
s^{1+o(1)}.
\)

For $p=o(1)$ and fixed $d$, condition~\eqref{eq:alpha-log-condition} is
equivalent to
\(
p\log N\to\infty.
\)
\end{corollary}

\begin{proof}
Choose
\begin{equation}
\label{eq:dense-nearlinear-tau}
\tau
:=
\left\lceil
\sqrt{\frac{\log N}{\alpha(p)}}
\right\rceil.
\end{equation}
Then
\(
\alpha(p)\tau
=
\Theta\left(
\sqrt{\alpha(p)\log N}
\right)
\to\infty.
\)
Theorem~\ref{thm:dense-tradeoff} therefore gives
\[
\frac{|E_{\mathrm{blk}}|}{\OPT(S)}
\le
1+
O\left(
\frac{1}{\sqrt{\alpha(p)\log N}}
\right)
=
1+o(1).
\]

For the running time,
\[
\frac{\tau}{\log N}
=
O\left(
\frac{1}{\sqrt{\alpha(p)\log N}}
\right)
=o(1).
\]
Hence
\[
N2^{O(\tau)}
=
N\exp(o(\log N))
=
N^{1+o(1)}.
\]

We now express the same running time in terms of the observed syndrome size.
On the syndrome-density event,
\(
N\le\frac{2s}{\alpha(p)}.
\)
Moreover,
\(
\alpha(p)\log N\to\infty
\)
implies
\[
\frac1{\alpha(p)}=o(\log N).
\]
Therefore
\[
s
\ge
\frac{N\alpha(p)}2
=
\frac{N}{N^{o(1)}},
\]
and hence
\(
\log s=(1-o(1))\log N.
\)
In particular,
\[
\exp(o(\log N))=s^{o(1)}
\]
and
\[
\frac{N}{s}
\le
\frac{2}{\alpha(p)}
=
N^{o(1)}
=
s^{o(1)}.
\]
Consequently,
\(
N2^{O(\tau)}
=
s^{1+o(1)}
\)
on the same event.
The success probability follows from~\eqref{eq:dense-tradeoff-prob}; since
$d=O(1)$, its failure term is $\exp(-\Omega(N\alpha(p)))$.
Finally, when $p=o(1)$,
Lemma~\ref{lem:dense-alpha} gives $\alpha(p)=\Theta(p)$, so
$\alpha(p)\log N\to\infty$ if and only if $p\log N\to\infty$.
\end{proof}

\begin{remark}[Equivalent parameter view]
The same conclusion can be expressed without the explicit choice
\eqref{eq:dense-nearlinear-tau}.  It is enough to choose any
\[
\tau=o(\log N)
\qquad\text{such that}\qquad
\alpha(p)\tau\to\infty.
\]
Such a choice exists exactly when $\alpha(p)\log N\to\infty$.
\end{remark}

\subsection{\texorpdfstring{Constant-\(p\) specialization}{Constant-p specialization}}
\label{subsec:dense-constant-p}

At constant noise rate the quantitative tradeoff takes its simplest form.

\begin{corollary}[Constant-rate approximation]
\label{cor:dense-constant-p}
Fix $p\in(0,1/2)$ and suppose every vertex is incident to exactly
$d=O(1)$ faces.  For every $\epsilon>0$, the block decoder can be parameterized
to output
\(
|E_{\mathrm{blk}}|
\le
(1+\epsilon)\OPT(S)
\)
with probability
\(
1-\exp(-\Omega(N)).
\)
Its deterministic running time is
\(
N2^{O(1/\epsilon)}.
\)
On the same high-probability event,
\(
s=\Theta(N),
\)
so the running time is equivalently
\(
s2^{O(1/\epsilon)}.
\)
\end{corollary}

\begin{proof}
For fixed $p$ and $d$,
\[
\alpha(p)=\frac{1-(1-2p)^d}{2}
\]
is a positive constant.  Theorem~\ref{thm:dense-tradeoff} with
\(
\tau=\Theta(1/\epsilon)
\)
therefore gives the claimed approximation, probability, and deterministic
running time.  Since $\alpha(p)$ is a positive constant,
Lemma~\ref{lem:dense-syndrome-tail} also gives $s=\Theta(N)$ with probability
$1-\exp(-\Omega(N))$, which yields the syndrome-sensitive form.
\end{proof}

\section{Transition Regime: Singleton Peeling and Residual Polymers}
\label{sec:transition}

The block approximation of Section~\ref{sec:dense} becomes asymptotically
near-optimal in near-linear time once $p\log N\to\infty$, whereas the exact
sparse decomposition of Section~\ref{sec:sparse} applies below the
$1/\log^2 N$ scale.  This section gives a different algorithm for the
intermediate range.

The key observation is that an isolated one-face physical error is visible
directly in the syndrome: all three vertices of that face are violated.  We
therefore peel every face having this local syndrome pattern.  Peeling need not
agree with a global optimum, and we make no localization claim of that kind.
Instead, we prove that every disagreement between the peeled set and the
physical error can be charged to a nearby pair of physical errors.  There are
only $O(Np^2)$ such pairs with high probability.  The residual ambiguity can
then be solved exactly in small syndrome components.

Throughout this section, $E\subseteq F$ denotes the iid generating error and
\(
S:=S(E),
\)
\(
s:=|S|,
\)
its observed syndrome and syndrome size.  When the syndrome is supplied as the
sparse list described in Section~\ref{sec:prelim}, the peeling step and the
construction of the residual problem can be performed without scanning the
entire lattice.  On the high-probability event used below, the residual
syndrome has size
\(
O(Np^2)=O(sp),
\)
so the algorithm spends its expensive exact-search work only on an additional
factor-$p$ fraction of the original syndrome.  This yields a successful-instance
running time $s^{1+o(1)}$, while the capped decoder retains deterministic
$N^{1+o(1)}$ worst-case time.

We use only the bounded-degree and two-dimensional growth assumptions from
Section~\ref{sec:prelim}.

\subsection{Singleton peeling rule}
\label{subsec:transition-peeling}

For a triangular face $f$, let $V(f)$ denote its three incident vertices.
Given only the syndrome $S$, define the \emph{peeled set}
\begin{equation}
\label{eq:peeled-set}
A=A(S)
:=
\{f\in F:V(f)\subseteq S\}.
\end{equation}
Thus $f$ is peeled exactly when all three of its vertices are violated.

For the analysis, define
\begin{equation}
\label{eq:peeling-discrepancy}
B:=E\triangle A.
\end{equation}
The decoder does not know $B$, since it does not know $E$.  Nevertheless,
linearity of the syndrome map gives
\begin{equation}
\label{eq:residual-syndrome}
S_{\mathrm{res}}
:=
S\triangle S(A)
=
S(E\triangle A)
=
S(B).
\end{equation}
Hence $B$ is a witness for the residual syndrome left after peeling.

We will measure local collisions using the face-adjacency metric.  Let $G_F$
be the graph whose vertices are the triangular faces and in which two faces are
adjacent when they share a vertex.  Write $\dist_F$ for graph distance in
$G_F$.  By bounded local geometry, every constant-radius ball in $G_F$ has
constant size, and every radius-$r$ ball has $O(r^2)$ faces.

Fix
\(
D:=2
\)
and define
\begin{equation}
\label{eq:nearby-pair-count}
X(E)
:=
\left|
\left\{
\{f,g\}\subseteq E:
f\neq g,\ \dist_F(f,g)\le D
\right\}
\right|.
\end{equation}
Any other sufficiently large fixed constant in place of $D$ would give an
equivalent quantity up to a constant factor.

\subsection{Charging peeling discrepancies to nearby error pairs}
\label{subsec:transition-charging}

The following deterministic lemma is the central reason peeling is useful.

\begin{lemma}[Peeling discrepancies require a collision]
\label{lem:transition-charge}
There is a lattice-dependent constant $C_{\mathrm{pair}}>0$ such that for
every face set $E$,
\begin{equation}
\label{eq:transition-charge}
|E\triangle A(S(E))|
\le
C_{\mathrm{pair}}X(E).
\end{equation}
\end{lemma}

\begin{proof}
Write
\[
B=E\triangle A
=
(E\setminus A)\sqcup(A\setminus E).
\]
We charge every face of $B$ to a pair counted by $X(E)$.

First consider
\(
f\in E\setminus A.
\)
Since $f\notin A$, at least one vertex $v\in V(f)$ is not in $S(E)$.  The
error face $f$ itself contributes one to the parity at $v$.  Since the total
parity at $v$ is even, there must be a second error face
\(
g\in E\setminus\{f\}
\)
incident to $v$.  Thus $f$ and $g$ share a vertex, so
\(
\dist_F(f,g)=1\le D.
\)
Charge $f$ to the error pair $\{f,g\}$.

Now consider
\(
f\in A\setminus E.
\)
Every vertex of $f$ belongs to the syndrome.  Therefore each of the three
vertices is incident to at least one error face.  A single triangular face
$g\neq f$ cannot contain all three vertices of $f$, since two distinct
triangular faces have different vertex sets.  Hence at least two distinct error
faces $g,h\in E$ are incident to vertices of $f$.  Both $g$ and $h$ are
adjacent to $f$ in $G_F$, so
\(
\dist_F(g,h)\le 2=D.
\)
Charge $f$ to one such pair $\{g,h\}$.

It remains to bound the number of charges received by one error pair
$\{g,h\}$.  Any face $f$ charged to $\{g,h\}$ lies within distance at most
one of $g\cup h$.  Indeed, if $f\in E\setminus A$, then the charging pair was
chosen to contain $f$ itself.  If $f\in A\setminus E$, then both members of
the charging pair are adjacent to $f$ in $G_F$.  Since the lattice has bounded
local degree, the radius-one neighborhood of $g\cup h$ contains only $O(1)$
faces.  Thus every error pair receives at most a lattice-dependent constant
number of charges.  Summing the charges proves~\eqref{eq:transition-charge}.
\end{proof}

\begin{remark}
Lemma~\ref{lem:transition-charge} does not assert that a peeled face belongs to
any global optimum.  Long-range rematching may make that false.  The lemma says
only that a wrong peeling decision requires a local collision of at least two
physical errors.
\end{remark}

\subsection{Concentration of the nearby-pair count}
\label{subsec:transition-pair-concentration}

The graph on faces that records possible charging pairs has bounded degree.
This gives an exponential upper tail for $X(E)$.

\begin{lemma}[Concentration of nearby error pairs]
\label{lem:transition-pair-concentration}
Assume
\(
Np^2\to\infty.
\)
Then there are constants $C_X,c_X>0$ such that
\begin{equation}
\label{eq:pair-tail}
\Pr\left(
X(E)>C_XNp^2
\right)
\le
\exp(-c_XNp^2).
\end{equation}
Consequently,
\(
|B|=O(Np^2)
\)
with probability $1-\exp(-\Omega(Np^2))$.
\end{lemma}

\begin{proof}
Let $G_D$ be the graph whose vertex set is $F$ and whose edges are the
unordered pairs $\{f,g\}$ satisfying
\(
0<\dist_F(f,g)\le D.
\)
Because $D$ is fixed and the lattice has bounded local degree,
\(
\Delta(G_D)=O(1)
\)
and
\(
|E(G_D)|=O(N).
\)

A bounded-degree graph has an edge coloring with a constant number
$q=O(1)$ of colors.  Write
\[
E(G_D)=M_1\sqcup\cdots\sqcup M_q,
\]
where every $M_j$ is a matching.  Let
\[
X_j
:=
\sum_{\{f,g\}\in M_j}
\mathbf 1[f\in E]\mathbf 1[g\in E].
\]
Because the edges in $M_j$ are vertex-disjoint, the summands in $X_j$ depend
on disjoint face-error indicators and are independent Bernoulli-$p^2$
variables.  Moreover,
\[
\mathbb E [X_j]
=
|M_j|p^2
=
O(Np^2).
\]

Choose a constant $C_0$ larger than twice the largest constant hidden in these
expectations.  A standard Chernoff upper-tail estimate gives
\[
\Pr(X_j>C_0Np^2)
\le
\exp(-\Omega(Np^2)).
\]
Since $q=O(1)$,
\(
X(E)=\sum_{j=1}^q X_j
\)
and a union bound proves~\eqref{eq:pair-tail}, after changing constants.
The final assertion follows immediately from
Lemma~\ref{lem:transition-charge}.
\end{proof}

We next relate the size of the observed syndrome to the number of nearby pairs of physical errors.

\begin{lemma}[Syndrome deficit from nearby error pairs]
\label{lem:syndrome-deficit}
There is a lattice-dependent constant $C>0$ such that for every
face set $E\subseteq F$,
\[
3|E|-|S(E)|
\le
C X(E).
\]
\end{lemma}

\begin{proof}
For every vertex $v$, let
\(
k_v:=|\{f\in E:v\in f\}|.
\)
Since every triangular face has three vertices,
\(
\sum_{v\in V}k_v=3|E|.
\)
Moreover,
\[
|S(E)|
=
\sum_{v\in V}(k_v\bmod 2).
\]
For every integer $k\ge0$,
\(
k-(k\bmod 2)
\le
2\binom{k}{2}.
\)
Therefore
\[
\begin{aligned}
3|E|-|S(E)|
=
\sum_{v\in V}
\left(k_v-(k_v\bmod 2)\right)
\le
2\sum_{v\in V}\binom{k_v}{2}.
\end{aligned}
\]
The sum on the right counts, with bounded multiplicity, pairs of distinct
error faces sharing a vertex. Every such pair has face distance one and is
therefore counted by $X(E)$. Hence
\[
3|E|-|S(E)|
\le
C X(E)
\]
for a lattice-dependent constant $C$.
\end{proof}

We next use these bounds to control the size of the residual syndrome after peeling.
\begin{lemma}[Residual syndrome size]
\label{lem:transition-residual-size}
For every error set $E$,
\(
|S_{\mathrm{res}}|
\le
3|B|.
\)
Hence, if $Np^2\to\infty$, then with probability
\(
1-\exp(-\Omega(Np^2))
\)
we have
\begin{equation}
\label{eq:transition-residual-size-N}
|S_{\mathrm{res}}|
=
O(Np^2).
\end{equation}
Moreover, under the additional condition $p=o(1)$, on an event of probability
$1-\exp(-\Omega(Np^2))$,
\begin{equation}
\label{eq:transition-syndrome-theta-Np}
s=\Theta(Np)
\end{equation}
and therefore
\begin{equation}
\label{eq:transition-residual-size-s}
|S_{\mathrm{res}}|
=
O(sp).
\end{equation}
\end{lemma}

\begin{proof}
By~\eqref{eq:residual-syndrome},
\(
S_{\mathrm{res}}=S(B).
\)
Every face has three vertices, so
\(
|S(B)|\le3|B|.
\)
Lemma~\ref{lem:transition-pair-concentration} then gives
\eqref{eq:transition-residual-size-N}.

For the second assertion, Lemma~\ref{lem:syndrome-deficit} gives
\[
3|E|-|S(E)|\le C X(E).
\]
On the pair-concentration event,
\(
X(E)=O(Np^2).
\)
Also, since $|F|=\Theta(N)$ and $Np^2\to\infty$ implies $Np\to\infty$,
Chernoff concentration gives
\(
|E|=\Theta(Np)
\)
with failure probability $\exp(-\Omega(Np))$, which is no larger than the
displayed $\exp(-\Omega(Np^2))$ scale after adjusting constants.  Therefore
\[
s
=
|S(E)|
=
3|E|-O(Np^2)
=
\Theta(Np)
\]
when $p=o(1)$.  Combining this with
\eqref{eq:transition-residual-size-N} gives
\(
|S_{\mathrm{res}}|=O(sp).
\)
\end{proof}

\subsection{The physical error is already near-optimal}
\label{subsec:transition-error-near-optimal}

The same collision count controls the difference between the physical error
weight and the true minimum correction weight.

Recall from Lemma \ref{lem:syndrome-deficit} that
\[
k_v
:=
|\{f\in E:v\in f\}|,
\qquad
\sum_{v\in V} k_v
=
3|E|,
\qquad
|S(E)|
=
\sum_{v\in V}(k_v\bmod 2).
\]

\begin{lemma}[Collision lower bound on the optimum]
\label{lem:transition-opt-lower}
There is a lattice-dependent constant $C_{\mathrm{opt}}>0$ such that, for
every error set $E$,
\begin{equation}
\label{eq:opt-vs-error-deterministic}
\OPT(S(E))
\ge
|E|-C_{\mathrm{opt}}X(E).
\end{equation}
Consequently, if $Np^2\to\infty$, then with probability
$1-\exp(-\Omega(Np^2))$,
\begin{equation}
\label{eq:opt-vs-error-random}
\OPT(S(E))
\ge
|E|-O(Np^2).
\end{equation}
\end{lemma}

\begin{proof}
By Lemma~\ref{lem:syndrome-deficit},
\[
|S(E)|
\ge
3|E|-C X(E).
\]
Every correction realizing $S(E)$ has at least $|S(E)|/3$ faces. Therefore
\[
\OPT(S(E))
\ge
\frac{|S(E)|}{3}
\ge
|E|-\frac{C}{3}X(E).
\]
The random statement follows
from Lemma~\ref{lem:transition-pair-concentration}.
\end{proof}

Thus the physical error itself is already a
$1+O(p)$ approximation to the optimum whenever $|E|=\Theta(Np)$ and
$X(E)=O(Np^2)$.  The remaining algorithmic problem is that $E$ is hidden.
Peeling recovers all but an $O(Np^2)$ ambiguous part, and we now show that this
residual can be solved exactly in near-linear time.

\subsection{Residual polymers}
\label{subsec:transition-polymers}
The following decomposition is used only in the analysis and is not
constructed by the decoder.  It is defined in terms of the hidden generating
error $E$.

Fix a sufficiently large constant
\(
D_0\ge 4.
\)
Let $G_{D_0}$ be the graph on all faces $F$ in which two distinct faces
$f,g$ are adjacent whenever
\[
\dist_F(f,g)\le D_0.
\]
The \emph{polymer graph} is the subgraph of $G_{D_0}$ induced by the error
faces $E$.  Its connected components are the \emph{error polymers}.
For a polymer $P$, its mass is
\(
|P|.
\)
We choose $D_0$ once and for all large enough to absorb all constant-radius
neighborhoods appearing in the charging argument.  In particular, it has the
following useful consequence.

\begin{lemma}[Residual components are supported near non-singleton polymers]
\label{lem:transition-B-polymers}
The constant $D_0$ can be chosen so that the following statements hold for
every error set $E$.

\begin{enumerate}[label=(\roman*)]
\item A singleton polymer contributes no face to $B=E\triangle A(S(E))$.

\item Every connected component $C$ of $B$ in the ordinary face-adjacency graph
$G_F$ can be associated with a unique non-singleton error polymer $P(C)$.

\item The component $C$ lies in a fixed constant-radius neighborhood of
$P(C)$, and
\begin{equation}
\label{eq:B-component-weight-polymer}
|C|
\le
C_B|P(C)|
\end{equation}
for a lattice-dependent constant $C_B$.
\end{enumerate}
\end{lemma}

\begin{proof}
Suppose first that $P=\{f\}$ is a singleton polymer.  No other error face lies
within distance $D_0$ of $f$.  In particular, no other error face is incident
to a vertex of $f$, so all three vertices of $f$ are in the syndrome and
$f\in A$.  Hence $f\notin E\setminus A$.
Moreover, no nearby face $g\neq f$ can be falsely peeled because of $f$.
Indeed, within the constant-radius neighborhood under consideration, $f$ is
the only error face that can contribute to the syndrome at vertices of $g$.
Since two distinct triangular faces share at most two vertices, at least one
vertex of $g$ is not incident to $f$, and hence is not in the syndrome.
Therefore $g\notin A$.
Thus a singleton polymer produces no face of $B=E\triangle A$, proving~(i).

In the proof of Lemma~\ref{lem:transition-charge}, every face of $B$ was
charged to an error pair at face distance at most $D=2$.  Such a pair lies in
one error polymer because $D_0\ge D$.  Associate the charged face with that
polymer.
If two adjacent faces of $B$ were associated with different polymers, then an
error face in the charging pair of the first and an error face in the charging
pair of the second would be separated by only a fixed constant in $\dist_F$.
Taking $D_0$ larger than this constant would join the two supposed polymers, a
contradiction.  Therefore every connected component $C$ of $B$ is associated
with a unique polymer $P(C)$.  This proves (ii), and the construction also
shows that $C$ lies in a constant-radius neighborhood of $P(C)$.

Finally, every face of $C$ is charged to a distance-$D$ error pair contained
in $P(C)$.  The graph whose vertices are error faces and whose edges join
pairs at face distance at most $D$ has bounded degree.  Therefore the
subgraph induced by the $|P(C)|$ faces of $P(C)$ contains only
$O(|P(C)|)$ such pairs.  Each pair receives only $O(1)$ charges by
Lemma~\ref{lem:transition-charge}.  Hence
\(
|C|=O(|P(C)|),
\)
proving~(iii).
\end{proof}

The point of using polymers is that after singleton peeling every relevant
polymer contains at least two physical errors.  This gains one power of the
coarse spatial scale in the counting argument below.

\subsection{Coarse polymer counting}
\label{subsec:transition-coarse-counting}

Set
\begin{equation}
\label{eq:transition-r}
r
:=
\left\lceil
A_0\frac{\log N}{\log(1/p)}
\right\rceil,
\end{equation}
where $A_0>0$ is a sufficiently large constant.  In the regime considered
below, $p\to0$, so this is well-defined for all sufficiently large $N$.

We first rule out a single polymer that stretches across distance comparable
to $r$.

\begin{lemma}[No long polymer]
\label{lem:transition-no-long-polymer}
For every fixed $\gamma>0$, the constant $A_0$ in
\eqref{eq:transition-r} can be chosen so that
\begin{equation}
\label{eq:no-long-polymer-tail}
\Pr\left(
\exists\text{ error polymer }P\text{ with }
\diam_F(P) \ge r/10
\right)
\le
N^{-\gamma}.
\end{equation}
for all sufficiently large $N$.
\end{lemma}

\begin{proof}
Because $P$ is connected in the polymer graph, any two faces of $P$ can be
joined by a sequence of error faces in which consecutive faces have
$\dist_F$-distance at most $D_0$.  Hence, if
\(
\diam_F(P)\ge r/10,
\)
then any such path joining a pair of faces at distance at least $r/10$ must
contain at least $r/(10D_0)$ edges.  Therefore
\(
|P|\ge c_0r
\)
for some constant $c_0>0$ depending only on $D_0$.

Since $G_{D_0}$ has bounded degree, there is a constant
$C_{\mathrm{conn}}>0$ such that, for any fixed face $f$, the number of
connected sets of $m$ faces in $G_{D_0}$ containing $f$ is at most
$C_{\mathrm{conn}}^m$. Since there are $O(N)$ choices for $f$, a union bound gives
\[
\Pr\left(
\exists\text{ polymer with }\diam_F(P)\ge r/10
\right)
\le
C N\sum_{m\ge c_0r}(C_\mathrm{conn}p)^m.
\]
For sufficiently large $N$, $C_{\mathrm{conn}}p<1/2$, and hence
\[
CN\sum_{m\ge c_0r}(C_{\mathrm{conn}}p)^m
\le
2CN(C_{\mathrm{conn}}p)^{c_0r}.
\]
Using~\eqref{eq:transition-r},
\begin{align*}
c_0r\log\frac1{C_{\mathrm{conn}}p}
\ge
c_0A_0\log N
\frac{\log(1/(C_{\mathrm{conn}}p))}{\log(1/p)}
=
(c_0A_0-o(1))\log N.
\end{align*}
Therefore
\begin{align*}
\Pr\left(
\exists\text{ polymer with }\diam_F(P)\ge r/10
\right)
&\le
2CN(C_{\mathrm{conn}}p)^{c_0r}\\
&=
2CN\exp\left(
-c_0r\log\frac1{C_{\mathrm{conn}}p}
\right)\\
&\le
2C N^{1-c_0A_0+o(1)}.
\end{align*}
Choosing $A_0$ so that
\(
c_0A_0>\gamma+2
\)
makes the last expression at most $N^{-\gamma}$ for all sufficiently
large $N$.
\eqref{eq:no-long-polymer-tail}.
\end{proof}

On the event of Lemma~\ref{lem:transition-no-long-polymer}, define an
interaction graph on the non-singleton polymers by connecting two polymers
whose face distance is at most $C_{\mathrm{int}}r$, where
$C_{\mathrm{int}}>0$ is a fixed geometric constant.
A connected component of this interaction graph is called a
\emph{coarse polymer collection}.  Its mass is the total number of physical
error faces in its polymers.

The next lemma is the main counting step.

\begin{lemma}[Coarse polymer mass bound]
\label{lem:transition-coarse-tail}
There is a constant $C_{\mathrm{pol}}>0$ such that for every integer $m\ge2$,
\begin{align}
&\Pr\left(
\begin{array}{c}
\text{there is no polymer of diameter at least }r/10,\\
\text{and there is a coarse polymer collection of mass at least }m
\end{array}
\right)
\nonumber
\le
C N\sum_{k\ge m}(C_{\mathrm{pol}}pr)^k.
\label{eq:coarse-tail}
\end{align}
In particular, whenever
\(
C_{\mathrm{pol}}pr<1/2,
\)
the right-hand side is
\(
O\left(
N(C_{\mathrm{pol}}pr)^m
\right).
\)
\end{lemma}

\begin{proof}
We first count collections of total mass exactly $k$.  Suppose such a
collection consists of $\ell$ non-singleton polymers
\(
P_1,\ldots,P_\ell
\)
with masses
\[
k_1,\ldots,k_\ell,
\qquad
k_1+\cdots+k_\ell=k.
\]
Since every polymer is non-singleton,
\begin{equation}
\label{eq:ell-at-most-half-k}
\ell\le k/2.
\end{equation}

Choose one representative error face from each polymer.  On the no-long-polymer
event, every polymer has face-metric diameter less than $r/10$.  Since the
polymer interaction graph is connected at range $C_{\mathrm{int}}r$, the
representatives form a connected set in a graph on all faces in which two
faces are adjacent when their face distance is at most $C'r$, for a fixed
constant $C'$.  By two-dimensional growth, this graph has maximum degree
\(
O(r^2).
\)
The number of rooted connected representative sets of size $\ell$ is therefore
at most
\begin{equation}
\label{eq:representative-count}
C N(C_2r^2)^{\ell-1}.
\end{equation}
One way to see this is to encode a connected set by a rooted spanning tree and
then expose its vertices in a tree traversal; each newly exposed representative
has $O(r^2)$ possible locations.

For fixed representatives and fixed polymer masses
$k_1,\ldots,k_\ell$, the number of choices for the internal connected shape of
the polymers is at most
\(
C_3^k
\)
for some constant $C_3>0$, because the polymer graph has constant degree.
The number of compositions of $k$ into positive parts is at most $2^k$.
Thus, the total number of candidate
configurations of $k$ specified error faces is at most
\[
C N(C_2r^2)^\ell (2C_3)^k.
\]
Every specified configuration requires all of its $k$ faces to be erroneous,
an event of probability $p^k$.  We do not need to enforce that the specified
sets are maximal polymers or that intervening faces are error-free; dropping
those conditions only enlarges the event being counted.

Using~\eqref{eq:ell-at-most-half-k},
\(
(r^2)^\ell
\le
r^k.
\)
After absorbing all fixed constants,
\[
\Pr\left(
\begin{array}{c}
\text{no long polymer, and some coarse collection}\\
\text{has total mass exactly }k
\end{array}
\right)
\le
C N(C_{\mathrm{pol}}pr)^k.
\]
Summing over $k\ge m$ completes the proof.
\end{proof}

We now specialize to the transition range.

\begin{lemma}[Maximum coarse mass]
\label{lem:transition-max-mass}
Assume
\begin{equation}
\label{eq:transition-upper-condition}
\frac{p\log N}{\log(1/p)}
\to0.
\end{equation}
Then $p\to0$ and
\(
pr=o(1).
\)
Fix any $\gamma>0$.  For a sufficiently large constant $K=K(\gamma)$, define
\begin{equation}
\label{eq:transition-M}
M
:=
\left\lceil
K\frac{\log N}
{\log(1/(C_{\mathrm{pol}}pr))}
\right\rceil.
\end{equation}
Then, with probability at least
\(
1-O(N^{-\gamma}),
\)
every coarse polymer collection has mass at most $M$ and no individual polymer
has diameter at least $r/10$.  Moreover,
\begin{equation}
\label{eq:M-small}
M=o(\log N).
\end{equation}
\end{lemma}

\begin{proof}
Condition~\eqref{eq:transition-upper-condition} first implies $p\to0$;
otherwise $\log(1/p)$ would remain bounded while $p\log N$ diverges along a
subsequence.  From~\eqref{eq:transition-r},
\[
pr
=
O\left(
\frac{p\log N}{\log(1/p)}
\right)
=
o(1).
\]
Thus
\begin{equation}
\label{eq:pr-upper}
pr=o(1),
\end{equation}
and in particular $C_{\mathrm{pol}}pr<1/2$ for all sufficiently large $N$.

By Lemma~\ref{lem:transition-coarse-tail}, the probability that a coarse
collection has mass at least $M$ while no long polymer occurs is at most
\[
C'N(C_{\mathrm{pol}}pr)^M.
\]
The definition of $M$ gives
\[
(C_{\mathrm{pol}}pr)^M
\le
N^{-K},
\]
so choosing $K>\gamma+2$ makes this probability $O(N^{-\gamma-1})$.
Choose $A_0$ in Lemma~\ref{lem:transition-no-long-polymer} so that the
long-polymer probability is also $O(N^{-\gamma-1})$.

Finally,~\eqref{eq:pr-upper} implies
\[
\log\frac1{C_{\mathrm{pol}}pr}\to\infty.
\]
For fixed $K$,
\[
\frac{M}{\log N}
=
O\left(
\frac{1}{\log(1/(C_{\mathrm{pol}}pr))}
\right)
=o(1),
\]
which proves~\eqref{eq:M-small}.
\end{proof}

At the upper endpoint $p=\Theta(1/\log N)$,
\[
r=\Theta\left(\frac{\log N}{\log\log N}\right),
\qquad
pr=\Theta(1/\log\log N),
\]
and therefore
\[
M
=
O\left(
\frac{\log N}{\log\log\log N}
\right).
\]
This is the worst case for the local search depth within the regime
$p=O(1/\log N)$.

\subsection{Syndrome-only residual components}
\label{subsec:transition-syndrome-components}

The polymer decomposition above is used only in the analysis.  The actual
decoder sees only $S_{\mathrm{res}}$.

Let $\dist_V$ denote shortest-path distance between vertices in the
underlying lattice graph.  By bounded
local geometry, face distance and vertex distance are equivalent up to absolute
constant factors on the scales used here.  Fix a sufficiently large constant
$C_S$ and construct a graph
\(
J_r
\)
whose vertex set is $S_{\mathrm{res}}$ and in which two residual syndrome
vertices are adjacent whenever
\begin{equation}
\label{eq:Jr-edge}
\dist_V(u,v)\le C_Sr.
\end{equation}
Let
\(
Q_1,\ldots,Q_\nu
\)
be the connected components of $J_r$.

We now show that every $Q_i$ is itself a valid syndrome with a small witness.
The proof uses connected components of the hidden set $B$, rather than the
individual error polymers.  This avoids any issue caused by cancellation of
syndrome vertices between nearby pieces.

\begin{lemma}[Residual syndrome components have small witnesses]
\label{lem:transition-small-witness}
On the event in Lemma~\ref{lem:transition-max-mass}, every component $Q_i$ of
$J_r$ is a valid syndrome.  More precisely, there are pairwise disjoint face
sets
\(
W_1,\ldots,W_\nu\subseteq B
\)
such that
\begin{equation}
\label{eq:Wi-syndrome}
S(W_i)=Q_i
\end{equation}
and
\begin{equation}
\label{eq:Wi-weight}
|W_i|\le C_WM
\end{equation}
for an absolute constant $C_W$.  In addition,
\begin{equation}
\label{eq:sum-Wi}
\sum_{i=1}^s |W_i|
\le
|B|.
\end{equation}
\end{lemma}

\begin{proof}
Decompose $B$ into connected components in the ordinary face-adjacency graph
$G_F$:
\(
B=C_1\sqcup\cdots\sqcup C_t.
\)
Distinct components have disjoint vertex supports, because two faces sharing a
vertex are adjacent in $G_F$.  Consequently their syndromes are disjoint and
\begin{equation}
\label{eq:B-syndrome-components}
S(B)
=
\bigsqcup_{j:S(C_j)\neq\varnothing} S(C_j).
\end{equation}

By Lemma~\ref{lem:transition-B-polymers}, each $C_j$ is associated with one
non-singleton polymer $P(C_j)$ and lies in a constant-radius neighborhood of
that polymer.  On the event of Lemma~\ref{lem:transition-max-mass}, every
polymer has face-metric diameter less than $r/10$.  Hence, after increasing
$C_S$ if necessary, every pair of vertices in $S(C_j)$ is at vertex distance
less than $C_Sr$.  Therefore all vertices of $S(C_j)$ lie in one component
$Q_i$ of $J_r$.  In particular, $J_r$ never splits the syndrome of one
connected component of $B$.

Now consider two $B$-components $C_j,C_{j'}$ whose syndrome vertices lie in
the same component $Q_i$.  Along a path in $J_r$ joining their syndrome
vertices, consecutive vertices are at distance $O(r)$.  If two consecutive
vertices belong to syndromes of different $B$-components, then the associated
polymers are also at face distance $O(r)$, since each $B$-component lies within
constant distance of its polymer.  Choosing $C_{\mathrm{int}}$ sufficiently
large, those polymers are adjacent in the coarse polymer interaction graph.
It follows that all polymers associated with $B$-components contributing to
one $Q_i$ lie in a single coarse polymer collection.

Let $\mathcal C_i$ be the set of $B$-components whose nonempty syndromes lie
in $Q_i$, and define
\[
W_i
:=
\bigcup_{C_j\in\mathcal C_i} C_j.
\]
The union is disjoint.  By~\eqref{eq:B-syndrome-components},
\[
S(W_i)
=
\bigsqcup_{C_j\in\mathcal C_i}S(C_j)
=
Q_i,
\]
so $Q_i$ is valid.

All associated polymers lie in a single coarse collection of total physical
error mass at most $M$.  By
Lemma~\ref{lem:transition-B-polymers},
\[
|C_j|
\le
C_B|P(C_j)|.
\]
Several $B$-components may be associated with the same polymer, but the proof
of that lemma actually charges all such faces to distance-$D$ pairs inside the
polymer; summing over all $B$-components associated with one polymer still
gives $O(|P|)$ faces.  Therefore
\[
|W_i|
\le
C_W
\sum_{P\text{ contributing to }Q_i}|P|
\le
C_WM.
\]
This proves~\eqref{eq:Wi-weight}.

Finally, the $\mathcal C_i$ partition the $B$-components having nonempty
syndrome.  Hence the $W_i$ are pairwise disjoint subsets of $B$, and
\(
\sum_i|W_i|\le|B|,
\)
which proves~\eqref{eq:sum-Wi}.
\end{proof}

\subsection{Exact residual solution}
\label{subsec:transition-residual-solver}

The decoder now solves each $Q_i$ independently by bounded search.  Let
$C_W$ be the constant in Lemma~\ref{lem:transition-small-witness}, and set
\(
K_{\mathrm{res}}:=\lceil C_WM\rceil.
\)
For every residual component $Q_i$, run the bounded-search algorithm from
Lemma~\ref{lem:bounded-search} with iterative deepening up to
$K_{\mathrm{res}}$.  If all searches succeed, let $D_i$ be the exact minimum
correction returned for $Q_i$ and output
\begin{equation}
\label{eq:transition-output}
C_{\mathrm{tr}}
:=
A\triangle D_1\triangle\cdots\triangle D_\nu.
\end{equation}

To make the running-time statement deterministic, if any residual search fails
within the prescribed budget, the algorithm aborts this branch and runs the
always-valid block decoder of Section~\ref{sec:dense} with a fixed constant
admissible wall spacing $\tau_0=O(1)$ on the original syndrome $S$.  This
fallback is used only on the exceptional event in the analysis; no
approximation claim is needed there.

\begin{lemma}[Residual solver]
\label{lem:transition-residual-solver}
On the event in Lemma~\ref{lem:transition-max-mass}, every bounded search
succeeds and
\begin{equation}
\label{eq:residual-solution-weight}
\sum_{i=1}^s |D_i|
\le
|B|.
\end{equation}
Moreover,
\(
S(D_1\triangle\cdots\triangle D_\nu)
=
S_{\mathrm{res}}.
\)
\end{lemma}

\begin{proof}
Lemma~\ref{lem:transition-small-witness} gives a witness $W_i$ for every
$Q_i$ with
\(
|W_i|\le C_WM\le K_{\mathrm{res}}.
\)
Hence bounded search reaches a feasible budget and returns an exact minimum
$D_i$.  Therefore
\(
|D_i|\le|W_i|.
\)
Summing and using~\eqref{eq:sum-Wi},
\[
\sum_i|D_i|
\le
\sum_i|W_i|
\le
|B|.
\]
The components $Q_i$ form a disjoint partition of $S_{\mathrm{res}}$, and
$S(D_i)=Q_i$.  By linearity,
\[
S(D_1\triangle\cdots\triangle D_\nu)
=
Q_1\triangle\cdots\triangle Q_\nu
=
S_{\mathrm{res}}. \hfill \qedhere
\]
\end{proof}

\begin{lemma}[Transition running times]
\label{lem:transition-runtime}
Assume
\(
\frac{p\log N}{\log(1/p)}\to0
\)
and use the parameters~\eqref{eq:transition-r} and
\eqref{eq:transition-M}.
The transition decoder, including its capped searches and fallback rule, has
deterministic running time
\(
N^{1+o(1)}.
\)
If the syndrome is supplied sparsely, then on the event
\(
|S_{\mathrm{res}}|=O(sp)
\)
together with the event of Lemma~\ref{lem:transition-max-mass}, the decoder does
not invoke the fallback and runs in
\begin{equation}
\label{eq:transition-runtime-sparse-general}
O\left(
s
+
|S_{\mathrm{res}}|r^2
+
|S_{\mathrm{res}}|\,d^{O(M)}\poly(M)
\right).
\end{equation}
If in addition $Np^2\to\infty$, this is
\(
s^{1+o(1)}.
\)
\end{lemma}

\begin{proof}
We first prove the deterministic cap.  The peeled set and residual syndrome can
always be computed by scanning the lattice in $O(N)$ time.  A radius-$O(r)$
neighborhood contains $O(r^2)$ lattice vertices by two-dimensional growth, so
$J_r$ can be constructed in
\(
O(Nr^2)
=
N^{1+o(1)}
\)
time.  Here $p\to0$ under the assumed condition, and hence
\[
r
=
O\left(\frac{\log N}{\log(1/p)}\right)
=
o(\log N).
\]

For a target $Q_i$, bounded search with budget
\(
K_{\mathrm{res}}=O(M)
\)
takes
\(
d^{O(M)}\poly(M+|Q_i|)
\)
time in the general formulation of Lemma~\ref{lem:bounded-search}.  The search
can be implemented with a mutable sparse syndrome set so that initialization
costs $O(|Q_i|)$ and every search-tree node incurs only $\poly(M)$ additional
work.  Thus its sharper implementation time is
\(
O\left(
|Q_i|+d^{O(M)}\poly(M)
\right).
\)
Lemma~\ref{lem:transition-max-mass} gives
\(
M=o(\log N),
\)
and therefore
\(
d^{O(M)}\poly(M)=N^{o(1)}.
\)
There are at most $N$ residual components, so the total capped search time is
$N^{1+o(1)}$.  If a search fails, the fixed-spacing block decoder with
admissible constant spacing $\tau_0$ takes $O(N)$ time by
Lemma~\ref{lem:dense-runtime}.  This proves the deterministic bound.

Now suppose the syndrome is supplied sparsely.  To construct $A$, enumerate
the at most $d$ faces incident to each vertex of $S$, deduplicate these
candidates, and test whether all three vertices of each candidate face belong
to $S$.  This takes $O(s)$ expected time with hashing, or an additional
logarithmic factor with a deterministic dictionary.  Moreover, every face in
$A$ contributes three incidences with vertices of $S$, while the total number
of such incidences is at most $ds$.  Hence
\(
3|A|\le ds,
\)
so
\(
|A|=O(s).
\)
The residual syndrome
\(
S_{\mathrm{res}}=S\triangle S(A)
\)
can therefore also be constructed in $O(s)$ time.

To construct $J_r$, explore a radius-$O(r)$ neighborhood around each residual
syndrome vertex and use a dictionary for membership in $S_{\mathrm{res}}$.
This costs
\(
O(|S_{\mathrm{res}}|r^2).
\)
The connected components of $J_r$ are then obtained in linear time in the
constructed graph.

Using the sparse bounded-search implementation above and
\(
\sum_{i=1}^{\nu}|Q_i|=|S_{\mathrm{res}}|,
\)
all exact residual searches together take
\(
O\left(
|S_{\mathrm{res}}|\,d^{O(M)}\poly(M)
\right).
\)
This proves~\eqref{eq:transition-runtime-sparse-general}.

Finally assume $Np^2\to\infty$.  By
Lemma~\ref{lem:transition-residual-size},
\(
|S_{\mathrm{res}}|=O(sp)
\)
on the pair-concentration event.  Also,
\(
d^{O(M)}\poly(M)=N^{o(1)}
\)
and $r^2=N^{o(1)}$.  Hence, the running time of the transition decoder satisfies
\[
T_{\mathrm{tr}}
\le
s+sp\,N^{o(1)}
=
sN^{o(1)}.
\]
It remains only to replace the $N^{o(1)}$ factor by $s^{o(1)}$.
The condition $Np^2\to\infty$ implies
\(
Np\gg \sqrt N.
\)
Together with
\(
s=\Theta(Np)
\)
from Lemma~\ref{lem:transition-residual-size}, this gives
\(
\log s=\Theta(\log N).
\)
Therefore
\(
N^{o(1)}=s^{o(1)},
\)
and so
\(
T_{\mathrm{tr}}=s^{1+o(1)}.
\)
\end{proof}

\subsection{Transition approximation theorem}
\label{subsec:transition-main}

We can now combine the peeling, residual, and optimum estimates.

\begin{theorem}[Transition-regime approximation]
\label{thm:transition-main}
Assume iid face noise with
\begin{equation}
\label{eq:transition-main-conditions}
Np^2\to\infty
\qquad\text{and}\qquad
\frac{p\log N}{\log(1/p)}\to0.
\end{equation}
Then the transition decoder has deterministic capped running time
\(
N^{1+o(1)}.
\)
If the syndrome is supplied sparsely, then on the same high-probability event
on which the approximation guarantee holds, its running time is
\(
s^{1+o(1)}.
\)
For every fixed $\gamma>0$, its constants can be chosen so that, with
probability at least
\begin{equation}
\label{eq:transition-success-prob}
1-
\exp(-\Omega(Np^2))
-
O(N^{-\gamma}),
\end{equation}
the decoder outputs a correction $C_{\mathrm{tr}}$ satisfying
\begin{equation}
\label{eq:transition-ratio}
|C_{\mathrm{tr}}|
\le
(1+O(p))\OPT(S(E)).
\end{equation}
Since~\eqref{eq:transition-main-conditions} implies $p\to0$, the approximation
ratio is $1+o(1)$.
\end{theorem}

\begin{proof}
The deterministic and syndrome-sensitive running times are given by
Lemma~\ref{lem:transition-runtime}.  We analyze the output on the intersection
of the following events:
\begin{enumerate}[label=(\roman*)]
\item
\(
X(E)=O(Np^2);
\)
\item no error polymer has diameter at least $r/10$; and
\item every coarse polymer collection has total error mass at most $M$.
\end{enumerate}
By Lemma~\ref{lem:transition-pair-concentration}, event (i) fails with
probability $\exp(-\Omega(Np^2))$.  By
Lemmas~\ref{lem:transition-no-long-polymer} and
\ref{lem:transition-max-mass}, the constants can be chosen so that (ii) or
(iii) fails with probability $O(N^{-\gamma})$.  This proves the success
probability~\eqref{eq:transition-success-prob}.

On these events, every residual bounded search succeeds.  By
Lemma~\ref{lem:transition-residual-solver},
\(
\sum_i|D_i|\le|B|.
\)
The output~\eqref{eq:transition-output} has the correct syndrome because
\begin{align*}
S(C_{\mathrm{tr}})
=
S(A)\triangle
S(D_1\triangle\cdots\triangle D_\nu)
=
S(A)\triangle S_{\mathrm{res}}
=
S(A)\triangle(S(E)\triangle S(A))
=
S(E).
\end{align*}

For its weight,
\begin{equation}
\label{eq:transition-output-step1}
|C_{\mathrm{tr}}|
\le
|A|+\sum_i|D_i|
\le
|A|+|B|.
\end{equation}
Since
\(
B=E\triangle A,
\)
we have
\(
A=E\triangle B
\)
and therefore
\(
|A|\le |E|+|B|.
\)
Substituting into~\eqref{eq:transition-output-step1},
\begin{equation}
\label{eq:transition-upper-vs-E}
|C_{\mathrm{tr}}|
\le
|E|+2|B|.
\end{equation}
Lemma~\ref{lem:transition-charge} and event (i) give
\(
|B|
=
O(Np^2),
\)
so
\begin{equation}
\label{eq:transition-output-Np2}
|C_{\mathrm{tr}}|
\le
|E|+O(Np^2).
\end{equation}
On the same event, Lemma~\ref{lem:transition-opt-lower} gives
\begin{equation}
\label{eq:transition-opt-Np2}
\OPT(S(E))
\ge
|E|-O(Np^2).
\end{equation}

It remains to compare $Np^2$ with $|E|$.  Since $|F|=\Theta(N)$ and each face
is erroneous independently with probability $p$,
\(
|E|\sim\operatorname{Bin}(|F|,p).
\)
The assumption $Np^2\to\infty$ implies $Np\to\infty$, so a Chernoff bound gives
\begin{equation}
\label{eq:E-concentration-transition}
|E|=\Theta(Np)
\end{equation}
with probability $1-\exp(-\Omega(Np))$.  This failure probability is smaller
than the displayed error terms in~\eqref{eq:transition-success-prob} after
adjusting constants.
Condition~\eqref{eq:transition-main-conditions} implies $p\to0$.  Combining
\eqref{eq:transition-output-Np2},
\eqref{eq:transition-opt-Np2}, and
\eqref{eq:E-concentration-transition},
\begin{align*}
\frac{|C_{\mathrm{tr}}|}{\OPT(S(E))}
\le
\frac{|E|+O(Np^2)}
{|E|-O(Np^2)}
=
\frac{1+O(p)}{1-O(p)}
=
1+O(p).
\end{align*}
This proves~\eqref{eq:transition-ratio}.
\end{proof}

\begin{corollary}[Logarithmic transition window]
\label{cor:transition-log-window}
For any fixed constants $c_0,c_1>0$, throughout
\[
\frac{c_0}{\log^2N}
\le
p
\le
\frac{c_1}{\log N},
\]
the transition decoder runs deterministically in
\(
N^{1+o(1)}
\)
time and outputs, with high probability, a
\(
1+O(p)=1+o(1)
\)
approximation.  With a sparse syndrome representation, its running time on the
same high-probability event is
\(
s^{1+o(1)}.
\)
\end{corollary}

\begin{proof}
The lower bound gives
\[
Np^2
\ge
\frac{c_0^2N}{\log^4N}
\to\infty.
\]
For the upper condition in Theorem~\ref{thm:transition-main},
\[
\frac{p\log N}{\log(1/p)}
\le
\frac{c_1}{\log(1/p)}
=
O\left(\frac1{\log\log N}\right)
\to0.
\]
Apply Theorem~\ref{thm:transition-main}.
\end{proof}

\begin{remark}[Why the syndrome-sensitive time is smaller]
The physical syndrome has typical size
\(
s=\Theta(Np),
\)
whereas after singleton peeling the residual syndrome has size only
\(
O(Np^2)=O(sp).
\)
Thus the exact part of the decoder is not applied to the original syndrome.
It is applied only to the collision-induced ambiguity, whose size is smaller by
a factor $p$.  The $s^{1+o(1)}$ running time is therefore not obtained merely
by substituting $N=\Theta(s/p)$ into an $N$-based bound.  Rather, the sparse
implementation avoids scanning the entire lattice: it examines only faces
incident to syndrome vertices and radius-$O(r)$ neighborhoods of residual
syndrome vertices.
\end{remark}

\begin{remark}[Why the effective parameter improves]
The advantage of singleton peeling is that every error cluster that remains
relevant to the residual problem contains at least two physical errors.

To see the effect on the counting argument, consider a coarse collection
consisting of $\ell$ polymers with total error mass $m$.  After choosing one
representative face from each polymer, the number of possible relative
locations of the polymers contributes a factor of order
\(
(r^2)^\ell,
\)
since a radius-$O(r)$ neighborhood in the two-dimensional lattice contains
$O(r^2)$ faces.  Because every relevant polymer is non-singleton,
\(
\ell\le \frac{m}{2},
\)
and therefore
\(
(r^2)^\ell
=
r^{2\ell}
\le
r^m.
\)

For any fixed configuration of $m$ faces, the probability that all of them
are erroneous is $p^m$.  Thus the geometric counting factor and the error
probability combine to give a term of order
\(
r^m p^m=(pr)^m.
\)
Without the non-singleton property, one could have $\ell=m$, leading instead
to the larger factor
\(
(r^2p)^m.
\)
Thus singleton peeling improves the effective expansion parameter from
$pr^2$ to $pr$.

At the upper end of the transition window,
\[
p=\Theta(1/\log N)
\qquad\text{and}\qquad
r=\Theta\left(\frac{\log N}{\log\log N}\right),
\]
so
\[
pr
=
\Theta\left(\frac1{\log\log N}\right)
\to0.
\]
This decay is what keeps the coarse residual collections small enough that
their syndromes can be solved exactly by bounded search in near-linear total
time.
\end{remark}

\section{Sparse Regime: Exact Decoding}
\label{sec:sparse}

We now turn to the regime
\(
p=o\left(\frac{1}{\log^2 N}\right).
\)
Here the goal is stronger than approximation: with high probability we recover
a globally minimum-size correction exactly.  The proof has two parts.  First,
the random syndrome decomposes into separated inflated active regions, each of
which has a small local witness.  Second, the annulus-crossing and
multi-region-crossing assumptions imply that a global optimum cannot escape
from one such region or merge several separated regions.  Exact global
decoding therefore reduces to exact local decoding.

The previously discussed local problems can be solved by bounded-treewidth dynamic
programming.  For the present regime, however, the bounded-search algorithm
from Lemma~\ref{lem:bounded-search} gives a sharper running time because the
local optimum itself is $o(\log N)$.

We write
\(
s:=|S|
\)
for the observed syndrome size.  With a sparse syndrome representation, only
$O(s)$ constant-scale blocks can be active.  Hence, on the high-probability
event on which localization succeeds, the decoder can be implemented without
scanning the entire lattice.  This yields a successful-instance running time
\(
(1+s)N^{o(1)},
\)
while the capped decoder retains deterministic $N^{1+o(1)}$ worst-case time.

Throughout this section the block size is a fixed constant.  Thus all constants
depending on the block scale are absorbed into the geometric constants from
Assumption~\ref{ass:geometry}.

\subsection{Inflated active components}
\label{subsec:sparse-active-components}

Let $S=S(E)$ be the syndrome generated by iid face noise.  A block
$B\in\mathcal B$ is called \emph{active} if it contains a syndrome vertex, in
the sense that
\(
S\cap V_0(\{B\})\neq\varnothing.
\)
Let
\(
\mathcal A\subseteq\mathcal B
\)
be the set of active blocks.

Recall from Section \ref{subsec:prelim-block-geometry} that $H_{\mathcal B}$ is the block-adjacency graph and
$\dist_{H_{\mathcal B}}$ denotes graph distance in $H_{\mathcal B}$.
Fix an integer radius $r\ge1$.  Define the \emph{coarse active graph}
$H^{(r)}$ as follows.  Its vertex set is $\mathcal A$, and two active blocks
$B,B'$ are adjacent whenever
\begin{equation}
\label{eq:sparse-active-edge}
\dist_{H_{\mathcal B}}(B,B')\le 6r.
\end{equation}
Let
\(
\mathcal K_r
\)
denote the set of connected components of $H^{(r)}$.
For a component $K\in\mathcal K_r$, let
\[
\mathcal N_{2r}(K)
:=
\{B\in\mathcal B:
\dist_{H_{\mathcal B}}(B,K)\le 2r\}
\]
be its radius-$2r$ neighborhood in the block-adjacency graph.  Define the
corresponding inflated face and vertex regions in the underlying lattice by
\[
F_{2r}(K):=F(\mathcal N_{2r}(K)),
\qquad
V_{2r}(K):=V(\mathcal N_{2r}(K)).
\]
Its local target syndrome is
\(
S_K:=S\cap V_{2r}(K).
\)

The intended local problem is
\begin{equation}
\label{eq:sparse-local-problem}
\min\left\{
|C_K|:
C_K\subseteq F_{2r}(K),\
S(C_K)=S_K
\right\}.
\end{equation}
Unlike the block problems of Section~\ref{sec:dense}, all syndrome parities are
enforced exactly here.  In particular, a local correction is not allowed to
leave residual syndrome on the boundary of its inflated region.

We will use the following capped version of the algorithm in order to make the
running-time guarantee deterministic.

\paragraph{Capped inflated active-component decoder.}
Fix once and for all an admissible constant wall spacing
\(
\tau_0=O(1)
\)
for the fallback block decoder. This parameter is
separate from the fixed constant block scale used throughout the present
section to define $\mathcal B$, $H_{\mathcal B}$, and the inflated regions.  

Given parameters $r,t,R$:
\begin{enumerate}[label=\arabic*.]
\item Construct the active blocks and the coarse graph $H^{(r)}$.
\item If any component $K\in\mathcal K_r$ has more than $t$ active blocks,
run the block decoder from Section~\ref{sec:dense} with wall spacing
$\tau_0$ on the original syndrome and terminate.
\item Otherwise construct $F_{2r}(K)$ and $S_K$ for every
$K\in\mathcal K_r$.
\item Run the bounded-search solver of Lemma~\ref{lem:bounded-search},
restricted to faces of $F_{2r}(K)$, with iterative deepening only up to
budget $R$.  If any search fails, run the block decoder with wall spacing
$\tau_0$ on the original syndrome and terminate.
\item If all searches succeed, let $C_K$ be the exact local optimum returned
for $S_K$, and output
\[
C_{\mathrm{sp}}
:=
\bigtriangleup_{K\in\mathcal K_r} C_K.
\]
\end{enumerate}

The fallback is needed only outside the high-probability event analyzed below.
Its purpose is to ensure that the decoder always terminates within the stated
worst-case running time.  The exactness theorem concerns the event on which no
fallback occurs.

\subsection{The sparse good event}
\label{subsec:sparse-good-event}

We package the geometric properties needed for exactness into one event.

\begin{definition}[Sparse good event]
\label{def:sparse-good-event}
For parameters $r,t,R\ge1$, the event
\(
\mathcal G(r,t,R)
\)
is the conjunction of the following four properties.

\begin{enumerate}[label=(\arabic*)]
\item \textbf{Small active components.}
Every $K\in\mathcal K_r$ contains at most $t$ active blocks:
\(
|K|\le t.
\)

\item \textbf{Separated inflated regions and syndrome partition.}
For distinct $K,K'\in\mathcal K_r$,
\[
F_{2r}(K)\cap F_{2r}(K')=\varnothing
\]
and
\[
V_{2r}(K)\cap V_{2r}(K')=\varnothing.
\]
Moreover,
\begin{equation}
\label{eq:sparse-syndrome-partition}
S
=
\bigsqcup_{K\in\mathcal K_r}S_K.
\end{equation}

\item \textbf{Small local witnesses.}
For every $K\in\mathcal K_r$, there is
\(
L_K\subseteq F_{2r}(K)
\)
such that
\begin{equation}
\label{eq:sparse-local-witness}
S(L_K)=S_K,
\qquad
|L_K|\le R.
\end{equation}

\item \textbf{Localization of every global optimum.}
Every minimum-size correction $C^\star$ for $S$ has a decomposition
\begin{equation}
\label{eq:sparse-opt-decomposition}
C^\star
=
\bigtriangleup_{K\in\mathcal K_r} C^\star_K
\end{equation}
such that
\[
C^\star_K\subseteq F_{2r}(K)
\qquad\text{and}\qquad
S(C^\star_K)=S_K
\]
for every $K$.
\end{enumerate}
\end{definition}

Properties (1)--(3) follow from random sparsity and lattice-animal estimates.
Property (4) is deterministic once the separation and witness properties are
known and the witnesses are sufficiently cheaper than an annulus crossing.
We prove these statements separately.

\subsection{Lattice-animal bounds}
\label{subsec:sparse-animal-bounds}

It is convenient to compare faces using block distance.  For a face $f$, let
$\mathcal B(f)$ be the set of blocks incident to its vertices, and define
\[
\dist_{\mathcal B}(f,g)
:=
\min\{
\dist_{H_{\mathcal B}}(B,B'):
B\in\mathcal B(f),\
B'\in\mathcal B(g)
\}.
\]
Bounded locality implies that every face is incident to only $O(1)$ blocks.

For a scale $r$, let $\Gamma_r$ be the graph on all faces in which
$f$ and $g$ are adjacent when
\(
\dist_{\mathcal B}(f,g)\le C_0r,
\)
where $C_0$ is a sufficiently large fixed constant.  Since the block size is constant throughout this section, each block is
incident to only $O(1)$ faces.  By two-dimensional growth, the
radius-$O(r)$ neighborhood of any constant-size set of blocks contains
$O(r^2)$ blocks, and hence $O(r^2)$ faces.  Therefore
\begin{equation}
\label{eq:sparse-Gamma-degree}
\Delta(\Gamma_r)=O(r^2).
\end{equation}
A face set is called \emph{$r$-coarsely connected} if it is connected in
$\Gamma_r$.

\begin{lemma}[Coarse error-animal bound]
\label{lem:sparse-coarse-animal}
There is a constant $C_{\mathrm{coarse}}>0$ such that, for every integer
$m\ge1$,
\begin{equation}
\label{eq:sparse-coarse-animal-tail}
\Pr\left(
\begin{array}{c}
\text{there exists an $r$-coarsely connected}\\
\text{set of at least $m$ erroneous faces}
\end{array}
\right)
\le
C N
\sum_{j\ge m}
(C_{\mathrm{coarse}}r^2p)^j.
\end{equation}
In particular, if
\(
\eta:=C_{\mathrm{coarse}}r^2p<1/2,
\)
then the probability is at most
\begin{equation}
\label{eq:sparse-coarse-animal-simple}
2CN\eta^m.
\end{equation}
\end{lemma}

\begin{proof}
The graph $\Gamma_r$ has maximum degree $O(r^2)$ by
\eqref{eq:sparse-Gamma-degree}.  In a graph of maximum degree $D$, the number
of connected vertex sets of size $j$ containing a fixed root is at most
$(C D)^j$ for an absolute constant $C$.  Therefore, after choosing a root face
in $O(N)$ ways, the number of $r$-coarsely connected face sets of size $j$ is
at most
\(
C'N(C''r^2)^j.
\)
A specified set of $j$ faces is entirely contained in the iid error $E$ with
probability $p^j$.  A union bound gives
\(
CN(C_{\mathrm{coarse}}r^2p)^j
\)
for size exactly $j$.  Summing over $j\ge m$ proves
\eqref{eq:sparse-coarse-animal-tail}.  When $\eta<1/2$, the geometric series is
at most $2\eta^m$.
\end{proof}

We also need a finer bound for a single physically connected error component.

\begin{lemma}[Long connected error component]
\label{lem:sparse-long-error}
There are constants $C_{\mathrm{long}},c_{\mathrm{long}}>0$ such that
\begin{equation}
\label{eq:sparse-long-error-tail}
\Pr\left(
\begin{array}{c}
\text{some physically connected component of $E$} \\
\text{has block diameter at least $r/2$}
\end{array}
\right)
\le
C N(C_{\mathrm{long}}p)^{c_{\mathrm{long}}r}.
\end{equation}
Here physical connectedness means face adjacency by a shared vertex.
\end{lemma}

\begin{proof}
A physically connected face set whose block diameter is at least $r/2$
contains at least $c_{\mathrm{long}}r$ faces for a geometric constant
$c_{\mathrm{long}}>0$.  The physical face-adjacency graph has constant
degree.  Hence the number of connected $j$-face sets containing a fixed root
is at most $C_0^j$.  Choosing the root in $O(N)$ ways and summing over
$j\ge c_{\mathrm{long}}r$ gives
\[
CN\sum_{j\ge c_{\mathrm{long}}r}(C_0p)^j.
\]
For sufficiently small $p$, this is bounded by
\(
C'N(C_{\mathrm{long}}p)^{c_{\mathrm{long}}r}
\)
after modifying constants.
\end{proof}

We next relate these error-animal events to active components.

\begin{lemma}[Large active components force coarse error animals]
\label{lem:sparse-active-implies-error}
There is a constant $c_{\mathrm{act}}>0$ such that if
$H^{(r)}$ has a connected component containing at least $t$ active blocks,
then $E$ contains an $r$-coarsely connected set of at least
$c_{\mathrm{act}}t$ error faces.
\end{lemma}

\begin{proof}
For every active block $B$, choose one syndrome vertex
\(
v_B\in S\cap V_0(\{B\}).
\)
Since $v_B$ has odd error parity, it is incident to at least one error face;
choose one and call it $f_B$.

By bounded locality, there is a fixed constant $C_{\rm loc}'$ such that
$f_B$ lies within block distance at most $C_{\rm loc}'$ of $B$.  If
$B,B'$ are adjacent in $H^{(r)}$, then
\(
\dist_{H_{\mathcal B}}(B,B')\le 6r.
\)
Hence
\[
\dist_{\mathcal B}(f_B,f_{B'})
\le
6r+2C_{\rm loc}'
\le
C_0r,
\]
where $C_0$ was chosen in the definition of $\Gamma_r$ sufficiently large.
Thus the distinct faces among the chosen $f_B$ form an
$r$-coarsely connected error set.

A fixed error face is incident, through its three vertices and the
constant-size block decomposition, to only $O(1)$ blocks.  Hence it can be
chosen as $f_B$ for only $O(1)$ active blocks.  Therefore $t$ active blocks
produce at least $c_{\mathrm{act}}t$ distinct error faces for an absolute
constant $c_{\mathrm{act}}>0$.
\end{proof}

To construct local witnesses, decompose the generating error $E$ into its
physically connected face components.  Distinct components have disjoint
vertex supports, because two faces sharing a vertex are adjacent.  Thus their
syndromes are disjoint.

\begin{lemma}[Local generating witnesses]
\label{lem:sparse-local-generating-witness}
Suppose no connected component of $E$ has block diameter at least $r/2$.
Then every connected error component with nonempty syndrome is associated with
a unique $K\in\mathcal K_r$, is contained in $F_r(K)$ for all sufficiently
large $r$, and all components associated with the same $K$ form an
$r$-coarsely connected error set.

Consequently, if in addition no $r$-coarsely connected error set has more than
$R$ faces, then every $K\in\mathcal K_r$ has a witness
$L_K\subseteq F_{2r}(K)$ satisfying
\[
S(L_K)=S_K,
\qquad
|L_K|\le R.
\]
\end{lemma}

\begin{proof}
Let $P$ be a physically connected component of $E$ with nonempty syndrome.
Choose a syndrome vertex $v\in S(P)$.  The vertex $v$ belongs to an active
block, hence to some component $K\in\mathcal K_r$.
Because the block diameter of $P$ is less than $r/2$ and every face incident
to $v$ is within $O(1)$ block distance of the block containing $v$, the entire
component $P$ lies in $F_r(K)$ for all sufficiently large $r$.

Every other syndrome vertex of $P$ is also within block distance less than
$r$ of the first one.  Therefore any active block containing such a vertex is
within distance less than $6r$ of the first active block and must belong to the
same component $K$ of $H^{(r)}$.  Hence $P$ is associated with a unique
$K$.

Now fix $K$.  Let
\(
P_1,\ldots,P_m
\)
be all physically connected components of $E$ whose nonempty syndromes are
associated with $K$.  Each $P_j$ is contained in $F_r(K)$.  Since $K$ is
connected in $H^{(r)}$, a path of active blocks between any two active blocks
of $K$ produces, by the same witness-face argument as in
Lemma~\ref{lem:sparse-active-implies-error}, a sequence of generating
components whose mutual block separations are $O(r)$.  Thus
\(
P_1\cup\cdots\cup P_s
\)
is $r$-coarsely connected.

Define
\(
L_K:=P_1\cup\cdots\cup P_s.
\)
Error components with empty syndrome can be omitted.  Since distinct physical
components have disjoint vertex supports,
\[
S(L_K)
=
\bigsqcup_{j=1}^m S(P_j).
\]
Every syndrome vertex lies in exactly one nonempty-syndrome error component,
and by uniqueness that component is associated with the active component
containing the vertex.  Hence
\(
S(L_K)=S_K.
\)
Also $L_K\subseteq F_r(K)\subseteq F_{2r}(K)$.  If there is no
$r$-coarsely connected error set larger than $R$, then
\(
|L_K|\le R.
\)
\end{proof}

We can now state the probabilistic part of the sparse good event.

\begin{lemma}[Probabilistic sparse decomposition]
\label{lem:sparse-probabilistic-good}
Let
\(
\eta:=C_{\mathrm{coarse}}r^2p.
\)
Assume $\eta<1/2$.  Then there are constants
$c_1,c_2,C>0$ such that the probability that any of the first three properties
of Definition~\ref{def:sparse-good-event} fails is at most
\begin{equation}
\label{eq:sparse-good-failure}
C N\eta^{c_1t}
+
C N\eta^{R}
+
C N(C_{\mathrm{long}}p)^{c_2r}.
\end{equation}
The separation and syndrome-partition portion of property (2) is deterministic
once $r$ is larger than a lattice-dependent constant.
\end{lemma}

\begin{proof}
If an active component has at least $t$ blocks,
Lemma~\ref{lem:sparse-active-implies-error} produces an
$r$-coarsely connected error set of at least $c_{\mathrm{act}}t$ faces.
Lemma~\ref{lem:sparse-coarse-animal} therefore bounds this failure probability
by
\(
CN\eta^{c_1t}
\)
for $c_1=c_{\mathrm{act}}$, after changing constants.

Lemma~\ref{lem:sparse-long-error} bounds the probability of a physically
connected error component of diameter at least $r/2$ by
\(
CN(C_{\mathrm{long}}p)^{c_2r}.
\)
Condition on the complement of that event.  By
Lemma~\ref{lem:sparse-local-generating-witness}, failure of the witness
property with budget $R$ implies an $r$-coarsely connected error set containing
more than $R$ faces.  Lemma~\ref{lem:sparse-coarse-animal} bounds this by
\(
CN\eta^R.
\)

It remains to justify the deterministic separation assertion.  If
$K\neq K'$ are distinct components of $H^{(r)}$, every active block in $K$ is
at block distance more than $6r$ from every active block in $K'$.  If
$V_{2r}(K)$ and $V_{2r}(K')$ intersected, bounded locality would imply that the
two block sets $\mathcal N_{2r}(K)$ and $\mathcal N_{2r}(K')$ come within
$O(1)$ block distance.  For sufficiently large $r$, this would place an active
block of $K$ within distance less than $6r$ of an active block of $K'$, a
contradiction.  The same argument gives disjointness of the face regions.

Finally, every syndrome vertex lies in at least one active block.  If that
active block belongs to $K$, then the vertex lies in $V_{2r}(K)$.  Since these
vertex regions are disjoint, the sets $S_K$ form a disjoint partition of $S$.
\end{proof}

\subsection{No merging and no escape}
\label{subsec:sparse-no-merging}

The remaining good-event property concerns the unknown global optimum.  This
part of the argument is deterministic.

We first note a simple structural fact.

\begin{lemma}[Optimum components have nonempty syndrome]
\label{lem:sparse-no-empty-opt-component}
Let $C^\star$ be a minimum-size correction, and decompose it into connected
components in the face-adjacency graph:
\[
C^\star=C_1\sqcup\cdots\sqcup C_m.
\]
Then
\(
S(C_j)\neq\varnothing
\)
for every $j$.
\end{lemma}

\begin{proof}
Distinct face components have disjoint vertex supports.  If some nonempty
component $C_j$ had empty syndrome, then
\(
C^\star\setminus C_j
\)
would realize the same global syndrome with strictly smaller weight,
contradicting minimality.
\end{proof}

We now prove the localization statement carefully.  The incidence-graph
argument below avoids an implicit assumption that one connected correction
component must realize an entire local syndrome.

\begin{lemma}[No merging and no escape]
\label{lem:sparse-no-merging-escape}
Suppose properties (2) and (3) of
Definition~\ref{def:sparse-good-event} hold, and suppose
\begin{equation}
\label{eq:sparse-exchange-condition}
2R<\lambda_0 r,
\qquad
\lambda_0:=\min\{\lambda_{\rm ann},\lambda_{\rm multi}\}.
\end{equation}
Then property (4) also holds: every global optimum decomposes over the inflated
regions.
\end{lemma}

\begin{proof}
Let $C^\star$ be a global minimum-size correction and write
\[
C^\star=C_1\sqcup\cdots\sqcup C_m
\]
for its connected face components.  By
Lemma~\ref{lem:sparse-no-empty-opt-component}, every $C_j$ has nonempty
syndrome.

Construct a bipartite incidence graph $\mathcal I$.  One side consists of the
correction components $C_j$, and the other side consists of the active
components $K\in\mathcal K_r$.  Put an edge between $C_j$ and $K$ when
\[
S(C_j)\cap S_K\neq\varnothing.
\]
Since the sets $S_K$ partition the global syndrome and the syndromes of the
$C_j$ are disjoint, every correction component has degree at least one and
every syndrome vertex is represented by exactly one incidence edge at its
correction-component endpoint.

\paragraph{No merging.}
Suppose some connected component $\mathcal J$ of the incidence graph contains
\(
\ell\ge2
\)
distinct active regions and $q$ correction components.  Let $e$ be its number
of incidence edges.  Since $\mathcal J$ is connected,
\begin{equation}
\label{eq:incidence-edge-count}
e\ge \ell+q-1.
\end{equation}
For a correction component $C_j$ in $\mathcal J$, let $d_j$ be its degree in
$\mathcal J$.

If $C_j$ is incident to $d_j$ active regions, then it contains a face incident
to a syndrome vertex in each of those regions.  Because every syndrome vertex
of $S_K$ lies within a bounded block distance of an active block in $K$,
bounded locality implies that $C_j$ intersects $F_r(K)$ for every adjacent
region $K$, once $r$ is sufficiently large.  The regions are pairwise
separated by more than $6r$, so the multi-region crossing property of Assumption \ref{ass:geometry} gives
\[
|C_j|
\ge
\lambda_{\rm multi}r(d_j-1)
\]
for $d_j\ge2$.  The same inequality is trivially true for $d_j=1$, since the
right-hand side is zero.

Let
\[
U:=\bigcup_{C_j\in\mathcal J} C_j.
\]
Summing the previous bound,
\begin{align}
|U|
&\ge
\lambda_{\rm multi}r
\sum_{C_j\in\mathcal J}(d_j-1)
\nonumber\\
&=
\lambda_{\rm multi}r(e-q)
\nonumber\\
&\ge
\lambda_{\rm multi}r(\ell-1),
\label{eq:no-merging-lower}
\end{align}
where the final inequality uses~\eqref{eq:incidence-edge-count}.

Let $K_1,\ldots,K_\ell$ be the active regions in $\mathcal J$.  By definition
of a connected component of the incidence graph, $U$ realizes exactly the
syndrome
\(
S_{K_1}\sqcup\cdots\sqcup S_{K_\ell}.
\)
Indeed, if a correction component in $U$ had syndrome in another active
region, that region would also lie in $\mathcal J$; conversely, every syndrome
vertex in one of the $K_i$ belongs to a correction component adjacent to that
region and hence lying in $\mathcal J$.

The local witnesses therefore provide an alternative realization
\(
L_{K_1}\triangle\cdots\triangle L_{K_\ell}
\)
of the same syndrome, with total weight at most
\(
\ell R.
\)
For $\ell\ge2$,
\(
\ell\le2(\ell-1),
\)
so~\eqref{eq:sparse-exchange-condition} gives
\begin{align*}
\ell R
\le
2R(\ell-1)
<
\lambda_0r(\ell-1)
\le
\lambda_{\rm multi}r(\ell-1)
\le
|U|.
\end{align*}
Replacing $U$ inside $C^\star$ by the local witnesses strictly decreases the
weight while preserving the global syndrome, a contradiction.  Therefore
every connected component of the incidence graph contains exactly one active
region.  In particular, no connected optimum component contributes syndrome
to two distinct active regions.

\paragraph{No escape.}
Fix one active region $K$, and let
\(
U_K
\)
be the union of all connected components of $C^\star$ whose syndrome lies in
$S_K$.  By the no-merging conclusion,
\(
S(U_K)=S_K.
\)

Suppose some connected component $C_j\subseteq U_K$ is not contained in
$F_{2r}(K)$.  Since $S(C_j)$ is a nonempty subset of $S_K$, the component
contains a face incident to a syndrome vertex associated with an active block
of $K$.  By bounded locality, $C_j$ intersects $F_r(K)$ for all sufficiently
large $r$.  Since it also exits $F_{2r}(K)$, the annulus-crossing property
gives
\(
|C_j|
\ge
\lambda_{\rm ann}r
\ge
\lambda_0r.
\)
Therefore
\(
|U_K|\ge\lambda_0r.
\)
But the local witness $L_K$ realizes the same syndrome $S_K$ and has
\(
|L_K|\le R<\lambda_0r
\)
by~\eqref{eq:sparse-exchange-condition}.  Replacing all of $U_K$ by $L_K$
would strictly decrease the weight of $C^\star$, again a contradiction.

Hence every optimum component associated with $K$ is contained in
$F_{2r}(K)$.  Define
\(
C^\star_K:=U_K.
\)
The regions are pairwise disjoint, and
\(
S(C^\star_K)=S_K.
\)
Thus
\(
C^\star
=
\bigtriangleup_{K\in\mathcal K_r}C^\star_K,
\)
which is property (4).
\end{proof}

Combining the probabilistic decomposition with the deterministic exchange
argument gives the parameterized good-event theorem.

\begin{theorem}[Parameterized sparse good event]
\label{thm:sparse-parameterized-good}
Let
\(
\eta:=C_{\mathrm{coarse}}r^2p.
\)
Suppose
\[
\eta<1/2,
\qquad
N\eta^{c_1t}=o(1),
\qquad
N\eta^R=o(1),
\qquad
N(C_{\mathrm{long}}p)^{c_2r}=o(1),
\]
and
\(
2R<\lambda_0r.
\)
Then
\(
\Pr(\mathcal G(r,t,R))=1-o(1).
\)
More quantitatively, its failure probability is bounded, up to a constant
factor, by
\[
N\eta^{c_1t}
+
N\eta^R
+
N(C_{\mathrm{long}}p)^{c_2r}.
\]
\end{theorem}

\begin{proof}
Lemma~\ref{lem:sparse-probabilistic-good} proves properties (1)--(3) with the
stated failure probability.  On their intersection,
Lemma~\ref{lem:sparse-no-merging-escape} proves property (4).
\end{proof}

\subsection{Exactness conditioned on the good event}
\label{subsec:sparse-exact-good}

\begin{lemma}[Exactness on the sparse good event]
\label{lem:sparse-exact-good}
If $\mathcal G(r,t,R)$ holds, the capped inflated active-component decoder does
not invoke its fallback and outputs a global minimum-size correction.
\end{lemma}

\begin{proof}
Property (1) ensures that no active component exceeds the cap $t$.
Property (3) gives, for every $K\in\mathcal K_r$, a feasible local witness
$L_K$ of weight at most $R$.  Hence bounded search succeeds within budget $R$
and returns an exact optimum $C_K$ for the local problem
\eqref{eq:sparse-local-problem}.  Thus no fallback is invoked.

By property (2), the local syndromes $S_K$ partition $S$ and the inflated
regions are disjoint.  Therefore the output
\(
C_{\mathrm{sp}}
=
\bigtriangleup_K C_K
\)
is feasible:
\[
S(C_{\mathrm{sp}})
=
\bigtriangleup_K S(C_K)
=
\bigtriangleup_K S_K
=
S.
\]
Also, because the supports are disjoint,
\(
|C_{\mathrm{sp}}|
=
\sum_K|C_K|.
\)

Let $C^\star$ be a global optimum.  Property (4) gives
\(
C^\star
=
\bigtriangleup_K C^\star_K,
\)
where
\[
C^\star_K\subseteq F_{2r}(K)
\qquad\text{and}\qquad
S(C^\star_K)=S_K.
\]
Since $C_K$ is an exact minimum for the same local problem,
\(
|C_K|\le|C^\star_K|
\)
for every $K$.  Hence
\[
|C_{\mathrm{sp}}|
=
\sum_K|C_K|
\le
\sum_K|C^\star_K|
=
|C^\star|
=
\OPT(S).
\]
Since $C_{\mathrm{sp}}$ is feasible, the reverse inequality is automatic.
Thus
\(
|C_{\mathrm{sp}}|=\OPT(S).
\)
\end{proof}

\subsection{Near-linear running time by bounded search}
\label{subsec:sparse-near-linear}

We now choose the parameters.  The choice below is tuned to obtain an
arbitrarily large fixed polynomial success probability while minimizing the
local search depth.

Define
\begin{equation}
\label{eq:sparse-Delta}
\Delta_N
:=
\min\left\{
\log\frac{1}{p\log^2N},
\log N
\right\}.
\end{equation}
Under
\(
p=o(1/\log^2N),
\)
we have
\(
\Delta_N\to\infty.
\)

Choose constants
\[
0<\rho<\lambda_0/4
\qquad\text{and}\qquad
\rho_t>0.
\]
For any desired fixed failure exponent $\gamma>0$, let $A=A(\gamma)>0$ be a
sufficiently large constant and set
\begin{equation}
\label{eq:sparse-parameter-choice}
r
:=
\left\lceil
A\frac{\log N}{\Delta_N}
\right\rceil,
\qquad
R:=\lfloor\rho r\rfloor,
\qquad
t:=\lceil\rho_t r\rceil.
\end{equation}
Since $\Delta_N\le\log N$, we have $r\ge A$, so $R\ge1$ after increasing $A$
if necessary.  Also
\[
\frac{r}{\log N}
=
O\left(\frac{1}{\Delta_N}\right)
=o(1),
\]
so the radii lie below the logarithmic geometric scale $L_N$ for all
sufficiently large $N$.

We first verify the sparse-good-event conditions.

\begin{lemma}[Parameter estimates]
\label{lem:sparse-parameter-estimates}
Fix any $\gamma>0$.  Under $p=o(1/\log^2N)$, choose the constant
$A=A(\gamma)$ in~\eqref{eq:sparse-parameter-choice} sufficiently large.
Then there is a constant $c>0$ such that, for all sufficiently large $N$,
\begin{equation}
\label{eq:sparse-eta-log}
\log\frac1\eta
\ge
c\Delta_N,
\qquad
\eta=C_{\mathrm{coarse}}r^2p.
\end{equation}
Moreover,
\[
N\eta^{c_1t}\le N^{-\gamma},
\qquad
N\eta^R\le N^{-\gamma},
\qquad
N(C_{\mathrm{long}}p)^{c_2r}\le N^{-\gamma},
\]
and
\(
2R<\lambda_0r
\)
for all sufficiently large $N$.
\end{lemma}

\begin{proof}
Let
\(
q_N:=p\log^2N.
\)
Then $q_N=o(1)$ and
\(
\Delta_N
=
\min\{\log(1/q_N),\log N\}.
\)
Ignoring the harmless ceiling in $r$,
\[
r^2p
=
O\left(
\frac{\log^2N}{\Delta_N^2}p
\right)
=
O\left(
\frac{q_N}{\Delta_N^2}
\right).
\]

First suppose
\(
\Delta_N=\log(1/q_N)\le\log N.
\)
Then $q_N=e^{-\Delta_N}$, so
\(
\eta
\le
C\frac{e^{-\Delta_N}}{\Delta_N^2}.
\)
Hence
\[
\log\frac1\eta
\ge
\Delta_N+2\log\Delta_N-O(1)
\ge
\frac{\Delta_N}{2}
\]
for sufficiently large $N$.

Now suppose
\(
\Delta_N=\log N.
\)
Then
\(
q_N\le N^{-1},
\)
and therefore
\(
\eta
=
O\left(
\frac{1}{N\log^2N}
\right).
\)
Again,
\(
\log\frac1\eta
\ge
\frac{\Delta_N}{2}
\)
for all sufficiently large $N$.  This proves
\eqref{eq:sparse-eta-log}.

Since
\[
t=\Theta(r)
=
\Theta\left(
A\frac{\log N}{\Delta_N}
\right),
\]
we can now use that $\log \frac{1}{\eta} \geq \frac{\Delta_N}{2}$ to obtain
\(
t\log\frac1\eta
=
\Omega(A\log N).
\)
The hidden constant depends only on the lattice constants and on $\rho_t$.
Thus, by taking $A=A(\gamma)$ sufficiently large,
\[
N\eta^{c_1t}
=
\exp\left(
\log N-c_1t\log(1/\eta)
\right)
\le
N^{-\gamma}.
\]

Likewise,
\[
R=\Theta(r)
=
\Theta\left(
A\frac{\log N}{\Delta_N}
\right),
\]
with proportionality constant $\rho>0$.  Increasing $A$ further if necessary
gives
\(
N\eta^R\le N^{-\gamma}.
\)

For the long-error term, observe that
\[
\log\frac1p
=
\log\frac1{q_N}
+
2\log\log N
\ge
\Delta_N.
\]
After absorbing the fixed constant $C_{\mathrm{long}}$, for sufficiently large
$N$,
\[
\log\frac1{C_{\mathrm{long}}p}
\ge
\frac{\Delta_N}{2}.
\]
Therefore
\(
r\log\frac1{C_{\mathrm{long}}p}
=
\Omega(A\log N).
\)
Taking $A$ sufficiently large also gives
\(
N(C_{\mathrm{long}}p)^{c_2r}\le N^{-\gamma}.
\)

Finally,
\(
2R
\le
2\rho r
<
\lambda_0r
\)
because $\rho<\lambda_0/4$.
\end{proof}

Thus the good event has particularly strong probability.

\begin{corollary}[Sparse good-event probability]
\label{cor:sparse-good-probability}
Fix any $\gamma>0$.  Under $p=o(1/\log^2N)$, choose
$A=A(\gamma)$ in~\eqref{eq:sparse-parameter-choice} sufficiently large.
Then
\begin{equation}
\label{eq:sparse-good-probability}
\Pr(\mathcal G(r,t,R))
\ge
1-O(N^{-\gamma}).
\end{equation}
\end{corollary}

\begin{proof}
Apply Theorem~\ref{thm:sparse-parameterized-good} and
Lemma~\ref{lem:sparse-parameter-estimates}.
\end{proof}

We next verify the running time of the algorithm.

\begin{lemma}[Sparse running times]
\label{lem:sparse-runtime}
With the parameter choice~\eqref{eq:sparse-parameter-choice}, the capped
inflated active-component decoder runs on every input syndrome in
\begin{equation}
\label{eq:sparse-runtime}
N^{1+O(1/\Delta_N)}\polylog N
=
N^{1+o(1)}
\end{equation}
time.

If the syndrome is supplied sparsely, then on the good event
$\mathcal G(r,t,R)$ the fallback is not invoked and the running time is
\begin{equation}
\label{eq:sparse-runtime-s}
(1+s)N^{O(1/\Delta_N)}\polylog N
=
(1+s)N^{o(1)}.
\end{equation}
\end{lemma}

\begin{proof}
We first prove the deterministic cap.  The block graph has constant degree.
Scanning all blocks and inspecting radius-$6r$ neighborhoods constructs
$H^{(r)}$ in
\(
O(Nr^2)
=
N\polylog N
\)
time.

The algorithm immediately falls back if some active component contains more
than $t$ blocks.  Therefore, before any local exact search is started, every
processed component $K$ satisfies
\(
|K|\le t.
\)
Its inflated region is contained in the union of $|K|$ radius-$2r$ block
balls, so it has at most
\(
O(tr^2)
\)
blocks and faces up to fixed lattice constants.  Since $t=\Theta(r)$ and
$r=O(\log N)$, this is $\polylog N$.  Thus all explicitly processed local
regions have only polylogarithmic size.

For each $K$, the bounded-search solver is capped at depth $R$.  Its running
time is
\(
d^{O(R)}\poly(R+|F_{2r}(K)|).
\)
The number of local instances is at most $O(N)$, and all polynomial factors are
polylogarithmic.  Since
\(
R
=
O\left(
\frac{\log N}{\Delta_N}
\right),
\)
we have
\[
d^{O(R)}
=
\exp\left(
O\left(
\frac{\log N}{\Delta_N}
\right)
\right)
=
N^{O(1/\Delta_N)}.
\]
Hence the total capped-search cost is
\(
N^{1+O(1/\Delta_N)}\polylog N.
\)
If a cap is exceeded or a local search fails, the fixed-spacing block decoder
is invoked.  With any fixed admissible constant wall spacing
$\tau_0=O(1)$, Lemma~\ref{lem:dense-runtime} gives $O(N)$ time.  This proves
the deterministic bound.

Now assume the syndrome is supplied sparsely and the good event holds.  By the
constant-scale block geometry, every syndrome vertex is incident to only
$O(1)$ blocks.  Therefore the number of active blocks is
\(
O(s).
\)
We can generate these active blocks directly from the sparse syndrome in
$O(s)$ time up to dictionary overhead.

To construct $H^{(r)}$, inspect a radius-$6r$ block neighborhood around each
active block.  By quadratic growth, this costs
\(
O(sr^2).
\)
Connected components of the resulting coarse graph are found within the same
bound.

On the good event every active component contains at most $t$ blocks.  For
each component, explicitly constructing its inflated region costs
\(
O(tr^2)=\polylog N.
\)
There are at most $O(s)$ active components, so all local-region construction
costs
\(
O(s\,\polylog N).
\)

Finally, there are at most $O(s)$ local exact searches.  Each is capped at
depth $R$ and has only polylogarithmic local input size.  Thus the total search
cost is
\[
O\left(
s\,d^{O(R)}\polylog N
\right)
=
sN^{O(1/\Delta_N)}\polylog N.
\]
Adding the initialization overhead gives~\eqref{eq:sparse-runtime-s}.
\end{proof}

\subsection{Exact sparse theorem}
\label{subsec:sparse-main}

We can now state the final sparse-regime result.

\begin{theorem}[Exact decoding below the $1/\log^2N$ scale]
\label{thm:sparse-main}
Assume iid face noise with
\(
p=o\left(\frac1{\log^2N}\right).
\)
Let $\Delta_N$ be defined by~\eqref{eq:sparse-Delta}.  For every fixed
$\gamma>0$, choose the constant $A=A(\gamma)$ in
\eqref{eq:sparse-parameter-choice} sufficiently large.  Then the capped
inflated active-component decoder has deterministic running time
\[
N^{1+O(1/\Delta_N)}\polylog N
=
N^{1+o(1)}.
\]
Moreover, with probability at least
\(
1-O(N^{-\gamma}),
\)
it outputs a globally minimum-size correction:
\(
|C_{\mathrm{sp}}|
=
\OPT(S(E)).
\)
If the syndrome is supplied sparsely, then on the same event its running time is
\[
(1+s)N^{O(1/\Delta_N)}\polylog N
=
(1+s)N^{o(1)}.
\]
\end{theorem}

\begin{proof}
The deterministic and syndrome-sensitive running-time bounds are given by
Lemma~\ref{lem:sparse-runtime}.  By
Corollary~\ref{cor:sparse-good-probability}, the sparse good event
$\mathcal G(r,t,R)$ holds with probability
\(
1-O(N^{-\gamma}).
\)
On that event, Lemma~\ref{lem:sparse-exact-good} shows that the capped decoder
does not invoke its fallback and returns an exact global optimum.
\end{proof}

\begin{remark}[Success probability versus local-search depth]
The choice
\(
r=\Theta\left(\frac{\log N}{\Delta_N}\right)
\)
is sufficient for failure probability $O(N^{-\gamma})$ for any prescribed
fixed $\gamma>0$, while giving the sharper search factor
\(
N^{O(1/\Delta_N)}.
\)
If a superpolynomially small failure probability is desired, one may multiply
$r$, $R$, and $t$ by a slowly growing factor $g(N)\to\infty$ satisfying
\(
g(N)=o(\Delta_N).
\)
The same counting proof then gives a superpolynomially small failure
probability, while
\(
d^{O(R)}
=
N^{O(g(N)/\Delta_N)}
=
N^{o(1)}.
\)
Thus stronger tails are available without losing near-linearity, at the cost
of a larger subpolynomial factor.
\end{remark}

\begin{remark}[Why the threshold is $1/\log^2N$]
The sparse proof requires the coarse occupation parameter
\(
\eta=\Theta(r^2p)
\)
to tend to zero while $r$ is large enough to suppress a connected physical
error component crossing the separation scale.  The latter requirement is
roughly
\(
r\log(1/p)\gtrsim\log N,
\)
so
\[
r
\gtrsim
\frac{\log N}{\log(1/p)}.
\]
Substituting this minimal scale into $r^2p=o(1)$ gives the heuristic condition
\[
p
\left(
\frac{\log N}{\log(1/p)}
\right)^2
=o(1).
\]
The simpler sufficient condition
\(
p=o(1/\log^2N)
\)
used above makes this compatibility immediate.  The singleton-peeling method
of Section~\ref{sec:transition} is designed precisely to avoid this
single-error coarse-density barrier in the transition region.
\end{remark}

\section{Combined Guarantees}
\label{sec:combined}

We now combine the sparse exact decoder, the singleton-peeling decoder, and the
block approximation algorithm.  The purpose of this section is to give a
single asymptotic picture and to make explicit both the deterministic
lattice-size running time and the sharper syndrome-sensitive running time on
the high-probability iid instances.

Throughout this section, assume
\(
0<p\le p_0<\frac12
\)
for a fixed constant $p_0$, and assume every vertex is incident to exactly
$d=O(1)$ faces whenever the block decoder is used.  Then
\[
\alpha(p)
:=
\frac{1-(1-2p)^d}{2}
\]
satisfies
\begin{equation}
\label{eq:combined-alpha-p}
\alpha(p)=\Theta(p)
\end{equation}
uniformly over the range considered below.

We write
\(
s:=|S(E)|
\)
for the observed syndrome size.  All syndrome-sensitive running times assume
the sparse syndrome representation described in Section~\ref{sec:prelim}.

\subsection{Overlap of peeling and block decoding}
\label{subsec:combined-overlap}

The transition theorem of Section~\ref{sec:transition} applies under
\begin{equation}
\label{eq:combined-transition-conditions}
Np^2\to\infty
\qquad\text{and}\qquad
\frac{p\log N}{\log(1/p)}\to0.
\end{equation}
Under these conditions, Theorem~\ref{thm:transition-main} gives a
$1+O(p)$ approximation with high probability.  Its capped running time is
deterministically
\(
N^{1+o(1)},
\)
and on the same high-probability event its sparse-input running time is
\(
s^{1+o(1)}.
\)

The block decoder of Section~\ref{sec:dense} applies whenever
\(
\alpha(p)\log N\to\infty.
\)
In that case Corollary~\ref{cor:dense-nearlinear} gives a
$1+o(1)$ approximation with deterministic running time
\(
N^{1+o(1)}
\)
and, on the syndrome-density event,
\(
s^{1+o(1)}
\)
running time.

These two ranges overlap.  We use this overlap to define a single hybrid
decoder above the sparse exact regime.

\subsection{A hybrid decoder above the exact threshold}
\label{subsec:combined-hybrid}

Define
\begin{equation}
\label{eq:combined-h}
h(N):=\sqrt{\log\log N}.
\end{equation}
More generally, any function satisfying
\(
h(N)\to\infty
\)
and
\(
h(N)=o(\log\log N)
\)
would work.

For noise rates at or above the $1/\log^2N$ scale, define the hybrid decoder as
follows.

\begin{enumerate}[label=\arabic*.]

\item If
\begin{equation}
\label{eq:combined-peeling-branch}
p\log N\le h(N),
\end{equation}
run the singleton-peeling decoder.

\item If
\begin{equation}
\label{eq:combined-block-branch}
p\log N>h(N),
\end{equation}
run the block decoder with
\begin{equation}
\label{eq:combined-tau}
\tau
=
\left\lceil
\sqrt{\frac{\log N}{\alpha(p)}}
\right\rceil.
\end{equation}

\end{enumerate}

The threshold $h(N)$ is only an algorithm-selection threshold.  It has no
intrinsic significance.

\begin{lemma}[Peeling branch]
\label{lem:combined-peeling-branch}
Assume
\(
p=\Omega(1/\log^2N)
\)
and
\(
p\log N\le h(N).
\)
Then the hypotheses of Theorem~\ref{thm:transition-main} hold.  Consequently,
the peeling branch has deterministic capped running time
\(
N^{1+o(1)}
\)
and, with high probability, outputs a
\(
1+O(p)=1+o(1)
\)
approximation.  With a sparse syndrome representation, its running time on the
same high-probability event is
\(
s^{1+o(1)}.
\)
\end{lemma}

\begin{proof}
From
\(
p=\Omega(1/\log^2N),
\)
we obtain
\[
Np^2
=
\Omega\left(\frac{N}{\log^4N}\right)
\to\infty.
\]

Also,
\(
p\le\frac{h(N)}{\log N},
\)
and hence
\[
\log\frac1p
\ge
\log\log N-\log h(N)
=
(1-o(1))\log\log N.
\]
Therefore
\[
\frac{p\log N}{\log(1/p)}
\le
\frac{h(N)}
{(1-o(1))\log\log N}
=
O\left(
\frac1{\sqrt{\log\log N}}
\right)
\to0.
\]
Thus Theorem~\ref{thm:transition-main} applies.

Finally,
\[
p
\le
\frac{\sqrt{\log\log N}}{\log N}
=
o(1),
\]
so the explicit approximation ratio
\(
1+O(p)
\)
is $1+o(1)$.  The syndrome-sensitive running time is part of
Theorem~\ref{thm:transition-main}.
\end{proof}

\begin{lemma}[Block branch]
\label{lem:combined-block-branch}
Assume
\(
p\log N>h(N).
\)
Then the block decoder with~\eqref{eq:combined-tau} has deterministic running
time
\(
N^{1+o(1)}
\)
and, with probability
\(
1-\exp(-\Omega(N\alpha(p))),
\)
achieves approximation ratio
\(
1+o(1).
\)
On the same event, its running time is
\(
s^{1+o(1)}.
\)
\end{lemma}

\begin{proof}
By~\eqref{eq:combined-alpha-p}, there is a constant $c>0$ such that
\(
\alpha(p)\ge cp.
\)
Hence
\[
\alpha(p)\log N
\ge
c\,p\log N
>
c\,h(N)
\to\infty.
\]
Corollary~\ref{cor:dense-nearlinear} therefore applies.

More explicitly,
\[
\frac{|E_{\mathrm{blk}}|}{\OPT(S)}
\le
1+
O\left(
\frac{1}{\sqrt{\alpha(p)\log N}}
\right)
\le
1+
O\left(
\frac{1}{\sqrt{h(N)}}
\right)
=
1+o(1),
\]
and
\[
\frac{\tau}{\log N}
=
O\left(
\frac{1}{\sqrt{\alpha(p)\log N}}
\right)
=
o(1).
\]
Thus
\(
N2^{O(\tau)}
=
N^{1+o(1)}.
\)
The syndrome-sensitive bound
\(
s^{1+o(1)}
\)
is exactly the second running-time conclusion of
Corollary~\ref{cor:dense-nearlinear}.
\end{proof}

\subsection{Final combined theorem}
\label{subsec:combined-main}

\begin{theorem}[Combined average-case decoding guarantees]
\label{thm:combined-main}
Consider iid face noise on a lattice family satisfying the assumptions of
Section~\ref{sec:prelim}, with
\(
0<p\le p_0<1/2.
\)
Assume in addition that the lattice is $d$-regular, with $d=O(1)$, whenever the
block branch is used.

\begin{enumerate}[label=(\roman*)]

\item If
\(
p=o\left(\frac1{\log^2N}\right),
\)
then for every fixed $\gamma>0$ there is a choice of parameters for which the
capped sparse decoder has deterministic running time
\(
N^{1+o(1)}
\)
and outputs a globally minimum-size correction with probability
\(
1-O(N^{-\gamma}).
\)
With a sparse syndrome representation, its running time on the same event is
\(
(1+s)N^{o(1)}.
\)

\item If
\(
p=\Omega\left(\frac1{\log^2N}\right),
\)
then the hybrid decoder above has deterministic capped running time
\(
N^{1+o(1)}
\)
and, with high probability, outputs a correction $C$ satisfying
\(
|C|
\le
(1+o(1))\OPT(S(E)).
\)
With a sparse syndrome representation, its running time on the same
high-probability event is
\(
s^{1+o(1)}.
\)

\end{enumerate}
\end{theorem}

\begin{proof}
Part (i) is Theorem~\ref{thm:sparse-main}.
For part (ii), use the hybrid decoder above.  If
\(
p\log N\le h(N),
\)
Lemma~\ref{lem:combined-peeling-branch} applies.  If
\(
p\log N>h(N),
\)
Lemma~\ref{lem:combined-block-branch} applies.  These cases are exhaustive for
every $N$, and both give deterministic $N^{1+o(1)}$ capped running time,
successful-instance running time $s^{1+o(1)}$, and approximation ratio
$1+o(1)$ with high probability.
\end{proof}

\begin{remark}[Success probabilities]
\label{rem:combined-success-probabilities}
In the sparse exact regime, for every fixed $\gamma>0$ the parameters can be
chosen so that the failure probability is
\(
O(N^{-\gamma}).
\)
As discussed after Theorem~\ref{thm:sparse-main}, a superpolynomially small
failure probability can also be obtained by increasing the separation radius
by a slowly growing factor while keeping the local-search overhead
subpolynomial.

In the peeling branch, for every fixed $\gamma>0$ the constants can be chosen
so that the failure probability is
\(
\exp(-\Omega(Np^2))
+
O(N^{-\gamma}).
\)
In the block branch, the failure probability is
\(
\exp(-\Omega(N\alpha(p))).
\)
For small $p$ this is
\(
\exp(-\Omega(Np)),
\)
while for constant $p$ it is
\(
\exp(-\Omega(N)).
\)
\end{remark}

\begin{remark}[Using an upper bound on the noise rate]
\label{rem:combined-upper-bound-p}
The sparse exact decoder does not require exact knowledge of the noise rate.
Suppose instead that a deterministic upper bound
\(
p\le \bar p
\)
is available, where
\(
\bar p=o(1/\log^2N).
\)
Define
\[
\bar\Delta_N
:=
\min\left\{
\log\frac{1}{\bar p\log^2N},
\log N
\right\},
\]
and choose the sparse-decoder parameters using $\bar\Delta_N$ in place of
$\Delta_N$.  Since $\bar p\ge p$, this gives a conservative choice of the
separation radius and search caps.  In particular, the probabilistic bounds
used in the sparse analysis remain valid, and the decoder is still exact with
high probability in $N^{1+o(1)}$ time whenever
$\bar p=o(1/\log^2N)$.

Thus exact knowledge of $p$ is not needed in the sparse regime; a sufficiently
small asymptotic upper bound is enough.  An upper bound alone does not,
however, determine the dense-regime parameter choice, since the block-decoder
approximation guarantee requires lower control on the syndrome density, or
equivalently on $\alpha(p)$.
\end{remark}

\begin{remark}[Interpretation of the syndrome-sensitive bounds]
\label{rem:combined-syndrome-sensitive}
The $s$-dependent running times are not obtained merely by substituting the
typical relation
\(
s=\Theta(Np)
\)
into an $N$-based analysis.

In the sparse regime, the decoder begins from only the $O(s)$ blocks touched
by syndrome vertices and explores their local neighborhoods.  In the peeling
regime, the expensive exact-search stage is applied only to a residual syndrome
of size
\(
O(sp).
\)
Only the dense block decoder is intrinsically global; there the relation
\(
N=O\left(\frac{s}{\alpha(p)}\right)
\)
on the syndrome-density event converts its lattice-size running time into an
$s$-sensitive one.
\end{remark}

\begin{remark}[Asymptotic regime language]
\label{rem:combined-regime-language}
The two rows above should be read as statements about sequences $p=p(N)$
satisfying the displayed asymptotic hypotheses.  As usual, an arbitrarily
oscillating sequence need not satisfy either
\(
p=o(1/\log^2N)
\)
or
\(
p=\Omega(1/\log^2N).
\)
For every individual $N$, however, the algorithms themselves remain
well-defined, and the hybrid selection rule above chooses between peeling and
block decoding without a gap near the $1/\log N$ crossover.
\end{remark}

\section{Discussion}
\label{sec:discussion}

The results suggest that average-case decoding can become tractable for
different structural reasons at different noise scales.  At very low noise,
the global optimization problem itself localizes.  At intermediate noise,
localization is no longer exact, but most of the instance can still be
resolved by simple local rules, leaving only a much smaller ambiguous core.
At higher noise rates, such localization is no longer necessary because the
optimum is sufficiently large that a controlled additive loss becomes
multiplicatively negligible.  Thus the phase diagram should not be viewed as
three unrelated algorithms, but as three ways in which random structure can
remove the difficulty present in worst-case instances.

A second implication is that the ambient lattice size $N$ is not always the
most informative measure of input size.  When the syndrome is supplied
sparsely, both the sparse and transition algorithms can restrict their work to
neighborhoods determined by the observed syndrome, and in the transition
regime the expensive exact computation is applied only to the much smaller
collision-induced residual.  This suggests that syndrome-sensitive complexity
may be a useful viewpoint more generally for decoding problems in which the
physical system is large but the observed defect set is sparse.

Several questions remain open.  Most notably, it is unclear whether the
$1/\log^2 N$ scale for exact decoding reflects an inherent change in the
problem or only the present localization argument.  The current proof pays a
two-dimensional coarse-placement factor for every possible physical error.
A sharper characterization of which local configurations can actually force
nonlocal behavior may permit exact decoding at higher noise rates.

The peeling analysis suggests another possible direction.  Singleton peeling
reduces the unresolved part of the instance from the natural $O(Np)$ scale to
$O(Np^2)$ by showing that every mistake requires a local collision.  More
sophisticated local rules might identify larger classes of configurations and
push the unresolved density to still higher powers of $p$.  Understanding the
minimal local obstructions to such peeling rules could lead to stronger
algorithms in the transition regime.

Finally, it would be useful to determine how much of this picture survives
beyond the iid minimum-size setting considered here.  Natural extensions
include finite-range dependent or other locally correlated noise models,
decoding with prescribed logical or homological constraints, and explicit
query or streaming models in which the lattice is accessed only near the
observed syndrome.  The proofs developed here isolate the probabilistic and
geometric ingredients used by each regime, which may help identify which of
these extensions can be obtained without changing the underlying algorithms.

\newpage
\bibliographystyle{plain}
\bibliography{ref}
\end{document}